\documentclass[conference]{IEEEtran}
\usepackage{graphicx}
\usepackage{amsthm}
\usepackage{epsfig}
\usepackage{latexsym}
\usepackage{amsfonts}
\usepackage{here}
\usepackage{rawfonts}
\usepackage[latin1]{inputenc}
\usepackage[T1]{fontenc}
\usepackage{calc}
\usepackage{capitalgreekitalic}
\usepackage{url}
\usepackage{enumerate}
\usepackage{color}
\usepackage[tbtags]{amsmath}
\usepackage{amssymb}
\usepackage{upref}
\usepackage{epic,eepic}
\usepackage{times}
\usepackage{dsfont}
\usepackage{comment}
\usepackage{cite}
\usepackage{bbm} 
\usepackage{optidef}
\usepackage{empheq}
\usepackage{amsmath}
\usepackage{enumitem}
\usepackage{booktabs}
\usepackage{caption}
\usepackage{array}     % 增强表格功能
\usepackage{booktabs}  
\usepackage{amsmath}   % 数学公式支持
\usepackage{multirow}  % 如果需要合并行
\usepackage{epstopdf}

\usepackage[ruled,vlined,linesnumbered]{algorithm2e}

\usepackage{hyperref}
\newtheorem{theorem}{\bf Theorem}

\usepackage{dsfont}

\newlength{\aligntop}
\newlength{\alignbot}
 \makeatletter
\renewenvironment{align}{%
  \vspace{\aligntop}
  \start@align\@ne\st@rredfalse\m@ne
}{%
  \math@cr \black@\totwidth@
  \egroup
  \ifingather@
    \restorealignstate@
    \egroup
    \nonumber
    \ifnum0=`{\fi\iffalse}\fi
  \else
    $$%
  \fi
  \ignorespacesafterend%
  \vspace{\alignbot}\par\noindent
} \makeatother

\IEEEoverridecommandlockouts

\usepackage{algorithmic}
\usepackage{subfigure}
\begin{document}
\bstctlcite{IEEEexample:BSTcontrol}
%\clearpage
\title{Inverse-Reinforcement Learning Enabled Digital Twin for Intent-based Drone Networks}
\vspace{-20mm}

% author names and affiliations
% use a multiple column layout for up to three different
% affiliations

% \author{
% \IEEEauthorblockN{Jiahao Wang\IEEEauthorrefmark{1}, Ruimin Yang\IEEEauthorrefmark{1}, Hanzhi Yu\IEEEauthorrefmark{2}, Huaiyu Dai\IEEEauthorrefmark{3}, and Ye Hu\IEEEauthorrefmark{1}\IEEEauthorrefmark{2}\vspace{-0.02cm}}
  
% \IEEEauthorblockA{\IEEEauthorrefmark{1}\small Department of Industrial and Systems Engineering, University of Miami, Miami, FL, USA, 33146.\\
% } 

%\IEEEauthorblockA{\IEEEauthorrefmark{2}\small Department of Electrical Engineering, Princeton University, Princeton, NJ, USA.\\
  %The Future Network of Intelligence Institute, Chinese University of Hong Kong, Shenzhen, China. Email: \protect{mingzhec@princeton.edu}.}

% \IEEEauthorblockA{\IEEEauthorrefmark{2}\small  Department of Electrical
% and Computer Engineering, University of Miami, Miami, FL, USA, 33146.\\
% } 

% \IEEEauthorblockA{\IEEEauthorrefmark{3}\small Electrical and Computer Engineering Department, North Carolina State University, Raleigh, NC, USA, 27606.\\
% Emails: jxw2293@miami.edu, rmyang@miami.edu, hanzhiyu@miami.edu, Huaiyu\_Dai@ncsu.edu, yehu@miami.edu,}

\author{
Jiahao Wang, Ruimin Yang, Hanzhi Yu, Huaiyu Dai~\IEEEmembership{Fellow,~IEEE}, and Ye Hu~\IEEEmembership{Member,~IEEE}%
\thanks{J. Wang and R. Yang are with the Department of Industrial and Systems Engineering, University of Miami, Coral Gables, FL 33146, USA, Emails: jxw2293@miami.edu, rmyang@miami.edu.}%
\thanks{H. Yu is with the Department of Electrical and Computer Engineering,
University of Miami, Coral Gables, FL 33146, USA, Email: hanzhiyu@miami.edu.}%
\thanks{H. Dai is with the Electrical and Computer Engineering Department, North Carolina State University, Raleigh, NC 27606, USA, Email: Huaiyu\_Dai@ncsu.edu.}%
\thanks{Y. Hu is with the Department of Industrial and Systems Engineering and the Department of Electrical and Computer Engineering,
University of Miami, Coral Gables, FL 33146, USA, Email: yehu@miami.edu.}%
\thanks{A preliminary version of this work was presented at the IEEE INFOCOM 2025 - IEEE Conference on Computer Communications Workshops (INFOCOM WKSHPS) \cite{Yang2025INFOCOMWKSHPS}.}
}

\vspace{-0.9cm}

%\thanks{This work was supported in part by the US National Science Foundation under grant ECCS-2203214}

% add grants
% use for special paper notices
%\IEEEspecialpapernotice{(Invited Paper)}% make the title area
\maketitle
%\vspace{-12mm}
\begin{abstract}
\boldmath
{In this paper, the problem of the trajectory design for an intent-based drone operating in resource-constrained, dynamic wireless network environments is studied. In the considered model, the drone acts as a supplementary base station that navigates among ground user clusters to provide on-demand uplink data access. Given its intended application (e.g traffic monitoring), the drone base station (DBS) prioritizes serving certain clusters (e.g. high-risk highway sections).  A digital twin (DT) system, hosted on a central server, creates a virtual representation of the physical wireless network environment to simulate and predict related changes, in which case the DBS trajectory should also be adjusted. Then, the DT system suggests adjustments to DBS trajectories without guaranteed access to the underlying DBS intent (i.e., service priorities), as this intent evolves over time and cannot be updated to the DT system in a timely manner due to intermittent connectivity between the DBS and the DT server. Such adjustment is posed as an optimization problem whose goal is to find the trajectories with which the fraction of prioritized users served by the DBS is maximized. To solve this problem under unknown DBS intent and unpredictable environment changes, an inverse reinforcement learning (IRL) based DT actuation solution is proposed. Simulation results demonstrate that the proposed solution provides near-real-time, near-optimal trajectory adjustment, with approximately 85\% less performance loss across environmental changes, compared to traditional reinforcement learning based on-board DBS control. The DT framework also enhances drone network performance by up to 2.5 times, compared to standard drone networks where a DBS operates with its erroneous and delayed environmental sensing.}

 \end{abstract}

\begin{keywords}
Intent-based drone networks, trajectory design, digital twin, inverse reinforcement learning.
\end{keywords}

\renewcommand{\thefootnote}{\fnsymbol{footnote}}

\IEEEpeerreviewmaketitle
 \vspace{-0.2cm}  
\section{Introduction}
 \vspace{-0.1cm}
{Drones are emerging as a promising solution to meet the increasing wireless communication demands in remote areas, heavily-populated urban cities, and disaster-stricken regions \cite{10032258,10588641,9345802,Luong2021TWC,Cang2023TWC,Wang2023TWC}. However, effectively deploying drones to meet real-world demands remains challenging. Specifically, how to maneuver drones to deliver timely, effective, on-demand data access is a major challenge, especially given the highly unpredictable nature of user demands and the wireless environment. Driven by the increasing demand of drone data collection in many applications (e.g., smart farming, manufacturing, etc.), the drones are expected to take highly diverse, complex intents in serving various user devices \cite{Sun2021JSAC, Zhu2023TWC, Liu2024TWC,Liu2021JSAC,Feng2022TWC,11216397}, which further exacerbates this challenge.}
% 10588641

\vspace{-0.2cm}
\subsection{Existing methods}
Existing literature \cite{8718556,Zhang2021TWC,8237204,9919620,8737778,9348212,9709537,9322414} has proposed a number of trajectory design solutions that directly operate on the drones. Particularly, prior works \cite{8718556,Zhang2021TWC,8237204} apply Gaussian based methods to estimate channel quality, user needs, and obstacle positions so as to deploy drones for optimal service performance. The work in \cite{8718556} applies local Gaussian modifiers to predict channel conditions at various drone positions, enabling rapid updates to current drone base station (DBS) trajectory plans under limited onboard computational resources. The authors in \cite{Zhang2021TWC} use a weighted expectation maximization based learning approach with Gaussian mixture guided traffic modeling to predict user distribution and downlink traffic demand, so as to optimize UAV deployment accordingly. In \cite{8237204}, a Gaussian mixture model is used to represent the prior distribution of target presence in rivers and plan DBS trajectories to detect a single and stationary target. Despite their promising results, solutions in \cite{8718556,Zhang2021TWC,8237204} are mostly under-fitted to real-world dynamics, as they mainly serve average user needs within an estimated wireless environment. By contrast, the work in \cite{9919620} studies the trajectory optimization of energy constrained DBSs to maximize the capacity of DBS-aided networks using a double Q learning solution. The work in \cite{9348212} adopts the transfer reinforcement learning (RL) method to accelerate the searching of DBS trajectories in post-disaster scenarios, with unknown user distribution and geographical features. In \cite{9709537}, the problem of trajectory optimization in wireless-powered DBS networks is addressed using a multi-agent RL approach. However, these RL based solutions \cite{9919620, 8737778, 9348212,9709537}, are mostly over-fitted to their training environment, and cannot provide timely adjustment of drone trajectory in response to unexpected environmental changes (e.g. user requests, or wireless impediments that did not appear within the training wireless environment). While the work in \cite{9322414,Hu2021JSAC} studies the usage of meta learning for accelerated trajectory adaptation in dynamic environments, it still requires the DBS to repeatedly fly around the target service area and collect information of environmental changes, which delays its response to users' needs.

Recently, there has been significant interest in using digital twin (DT) technology for decision making in unpredictable complex, dynamic environments. By synchronizing critical environmental information into the digital dimension, DT provides a systematic framework for monitoring, simulating, and predicting dynamic wireless environments. Particularly, existing literature\cite{dt-without1,dt-without2,Zhou2026TMC} has explored the usage of DT in drone networks. The authors in \cite{dt-without1} propose a DT of drone swarms, within which a high-fidelity virtual replica of the physical drone swarm is constructed based on real-time sensing of the physical-world dynamics, to guide the planning of coordinated drone operations. The authors in \cite{dt-without2} study a federated DT framework that fuses cooperative sensing and federated aggregation of local DT models to improve real-time DT faithfulness with rapidly changing system states and network topologies. In \cite{Zhou2026TMC}, a cooperative multi-scale DT framework is proposed for low-altitude delivery networks, in which macro-scale DT derives delivery associations between mobile delivery groups and parcel clusters, while micro-scale DT enables real-time path planning from local sensing data. However, all these works \cite{dt-without1,dt-without2,Zhou2026TMC} focus on observing and modeling the physical environment within DT, without interpreting DBS intents, which vary across time and among DBSs. Meanwhile, they lack the perspective of capturing network dynamics driven by DBS decisions, and mainly rely on optimization-driven paradigms to search decisions over predicted dynamic. So they are generally computationally intensive when handling the complex interplay of DBS operations and real-world wireless dynamics. Some more recent work leverages generative model-based RL, to model decision-driven environmental dynamics and automate decision making in DT \cite{Schena2024JOCS,Imanberdiyev2016ICARCV,Chen2023MASS}. In \cite{Schena2024JOCS}, reinforcement twinning is used to improve control performance and sample efficiency by combining a DT model with model-based policy learning, under the strong assumption of deterministic system dynamics and consistency between the virtual and real environments. The authors in \cite{Imanberdiyev2016ICARCV} use a decision-tree to model the environmental transition dynamics to allow real-time adjustment of UAV navigation in highly dynamic environments. The work in \cite{Chen2023MASS} proposes a DT-assisted model-based RL framework for joint task offloading and bandwidth allocation in collaborative edge computing networks, in which a state transition model is learned to capture dynamic edge-node resource availability and network conditions. This design enables more adaptive and efficient resource management in collaborative edge computing environments. Despite their promising results \cite{Schena2024JOCS,Imanberdiyev2016ICARCV,Chen2023MASS}, these works predict solely network changes, without understanding DBS intents, i.e., their service priorities, the full access to which is not guaranteed in real-world DT applications due to intermittent connectivity between the DBS and the DT server, making them impractical in real world drone networks.

\vspace{-0.2cm}
\subsection{Contributions}

The main contribution of this paper is a novel DT framework that actuates intent-based drone maneuvering decisions in dynamic wireless environments. We innovatively build a generative inverse reinforcement learning (IRL) model that enables the DT to infer DBS intents, observe network changes, and autonomously actuate the adaptation of DBS trajectories across various environmental changes. In brief, our key contributions include:
\begin{itemize}
 %\vspace{-0.02cm}
\item We consider a practical drone-aided wireless system in which a DBS navigates, under strict energy constraints and limited situational awareness, to provide supplementary uplink connectivity to distributed ground users for an intended application.
We propose to use a DT system to monitor, model, and predict network changes (e.g., wind velocity, user data requests, and energy consumption), thereby actuating decisions of drone trajectories without requiring access to drone intent. 
\item We cast the DT trajectory actuation into an optimization problem setting where the drone intent is captured as maximizing the fraction of prioritized user clusters that are served under the energy constrained and wind impacted flights. The DT system detects and predicts the environmental changes that can affect this optimization objective, thereby adapting the drone trajectory accordingly. This eliminates the need for the DBS to continuously gather environmental information and retrain its navigation decisions to fulfill its service intent.  
\item A novel generative IRL based solution is proposed to solve the DT assisted DBS trajectory actuation problem. Particularly, this solution trains a policy function through the inference of DBS intents to replicate DBS optimal trajectories in similar wireless environments. Then, the trained policy, which is taken as the digital replica of the physical DBS, can be directly used to actuate trajectory decisions to fulfill drone intents upon environmental changes. The proposed approach utilizes generative adversarial networks (GANs) to synthesize DBS trajectories under various conditions to foster high-performance trajectory actuation across various environmental changes.
\end{itemize}
Our analytical and numerical results show that, as the DT system provides timely and accurate accurate prediction of the physical wireless environment, the drone network can reduce the degradation of its service coverage across environmental changes. Our numerical results justify that the proposed method can achieve near-optimal DBS service coverage benchmarked by exhaustive search. Specifically, by integrating the GAN model with the imitation learning mechanism, the proposed solution yields near-optimal service coverage, despite the lack of accurate information on DBS intents. They also verify that the generative learning model also helps the drone network adapt to environmental changes up to 2 times faster than the PPO algorithm, with 85\% less performance loss across the changes. Overall, it is verified that the proposed framework can offload the prediction of environmental changes and the adjustment of DBS trajectories to a DT of the drone network, which infers DBS intents and actuates the trajectory adjustment with higher scalability, reliability, and communication performance. 

The rest of this paper is organized as follows. The system model and problem formulation are described in Section II. In Section III, the proposed algorithm is discussed. In Section IV, numerical simulation results are presented. Finally, conclusions are drawn in Section V.

%\vspace{-0.12cm}

\section{System Model and Problem Formulation}
We consider an infrastructure-limited geographical area where a set $\mathcal{U}$ of $U$ randomly distributed terrestrial users are requesting uplink data access that cannot be served by terrestrial base stations (BSs). To satisfy these requests, a fixed-wing DBS is dispatched. As shown in Fig. \ref{Fig. 1}, this DBS flies around the target area to serve terrestrial users, and can temporarily connect to a central server wirelessly through the distributed BSs whose coverage does not overlap with this infrastructure-limited geographical area to get maneuvering recommendations. 
The terrestrial users are divided into several groups, each of which is called a cluster representing a geographical region with a radius of $d_r$ that the DBS can and must cover. The set of these clusters is denoted as $\mathcal{C}$. Notice that, based on its service intent, the DBS can have different service priorities for these clusters. For example, a DBS monitoring traffic will prioritize clusters with the highest vehicle density. We assume the weight that this DBS assigns to serving cluster $k\in\mathcal{C}$ is $w_{k}$, with $w_{k}$ capturing that the DBS puts a higher priority on serving this cluster.
Meanwhile, within these clusters, each user $u$ independently requests $b_u$ bits of data services at time $t_u$. We use $\boldsymbol{b} = [b_1, \dots, b_U]$ and $\boldsymbol{t} = [t_1, \dots, t_U]$ to, respectively, represent the vector of these independent user requests and the occurrence times. 

The DBS travels among the clusters, receives and fulfills their data uplink requests. 
Particularly, it flies among these clusters along steady straight-and-level flights (SLFs), and hovers over each cluster in a circular flight (a circle with radius $r$), with constant altitude $H$ and ground speed $V$. 
However, due to the influence of wind with velocity $\boldsymbol{\vec{v}_w} = v_w\angle\alpha$ where $v_w$ represents the wind speed in the horizontal plane, and $\alpha$ denotes the wind direction angle measured counterclockwise from the east direction, the actual flying velocity of the DBS will be $\boldsymbol{\vec{v}} = \boldsymbol{\vec{v}_g} + \boldsymbol{\vec{v}_w}$ where $\boldsymbol{\vec{v}_g}$ represents the ground velocity vector that the DBS uses against wind impacts. As the wind velocity varies over times and locations, the DBS must consume energy to maintain its flying velocity. The DBS is assumed to be myopic without prior knowledge about the actual user requests and wind velocity at each cluster before its arrival. Moreover, it flies under energy constraints such that it must return to its original $O$ within a time period $T$ for battery recharging. The DBS's movement can be represented by a trajectory vector $\boldsymbol{\xi} = [\xi_1, \xi_2, \dots, \xi_K]^\top$ where $\xi_k \in \mathcal{C} \cup \{O\}$ represents the $k$-th cluster that the DBS serves, or the initial location $O$ to which it returns after serving the users. $K$ represents the maximum number of clusters that the DBS can serve before returning.
In the considered wireless network, the DBS, BSs, and user clusters are considered as the \emph{physical network objects (PNOs)}. The status of each PNO can be described by its internal properties (e.g., DBS location, flying velocity, energy budget, service-intent weights, BS locations, user-cluster locations, user-cluster service requests, and wind velocity) as well as by its external properties (e.g., communication performance). The central server is not guaranteed access to weights $w_{k}$, due to the wide variety and dynamic nature of DBS intents.

The central server operates a DT system that functions as a real-time virtual representation of the physical wireless network, such that it is expected to comprise the same PNOs, with the same internal and external properties. Specifically, the DT system is built based on the following three processes: 
1) Mapping the PNOs into a virtual environment by replicating them as vectors of their internal and related external properties. As shown in Fig. \ref{Fig. 1}, the DT system collects PNO properties by periodically synchronizing related data with distributed sensing units.  
2) Evolving the internal and external properties of virtual PNOs. In particular, once a change of internal or external properties at one PNO is mapped to or predicted by the DT system, the internal and external properties of other PNOs will also be updated\footnote{Due to rapid environmental changes and communication delays, the DT does not have guaranteed access to DBS intents.}.
3) Actuating the DBS trajectory adaptation upon changes in virtual PNO properties, with the goal of achieving optimal prioritized service coverage. In what follows, we will first explain how the DT system evolves PNO properties concerning communication performance, flying velocity, and energy budget. The PNO locations, wind velocity, and user requests are inferred using data-driven sensing, synchronization, and prediction methods, following our previous work \cite{Tong2025TMC}. Then, we formulate the DBS trajectory actuation as an optimization problem. 
%\vspace{-0.2cm}

\begin{figure}
  \centering
  \includegraphics[width=0.5\textwidth]{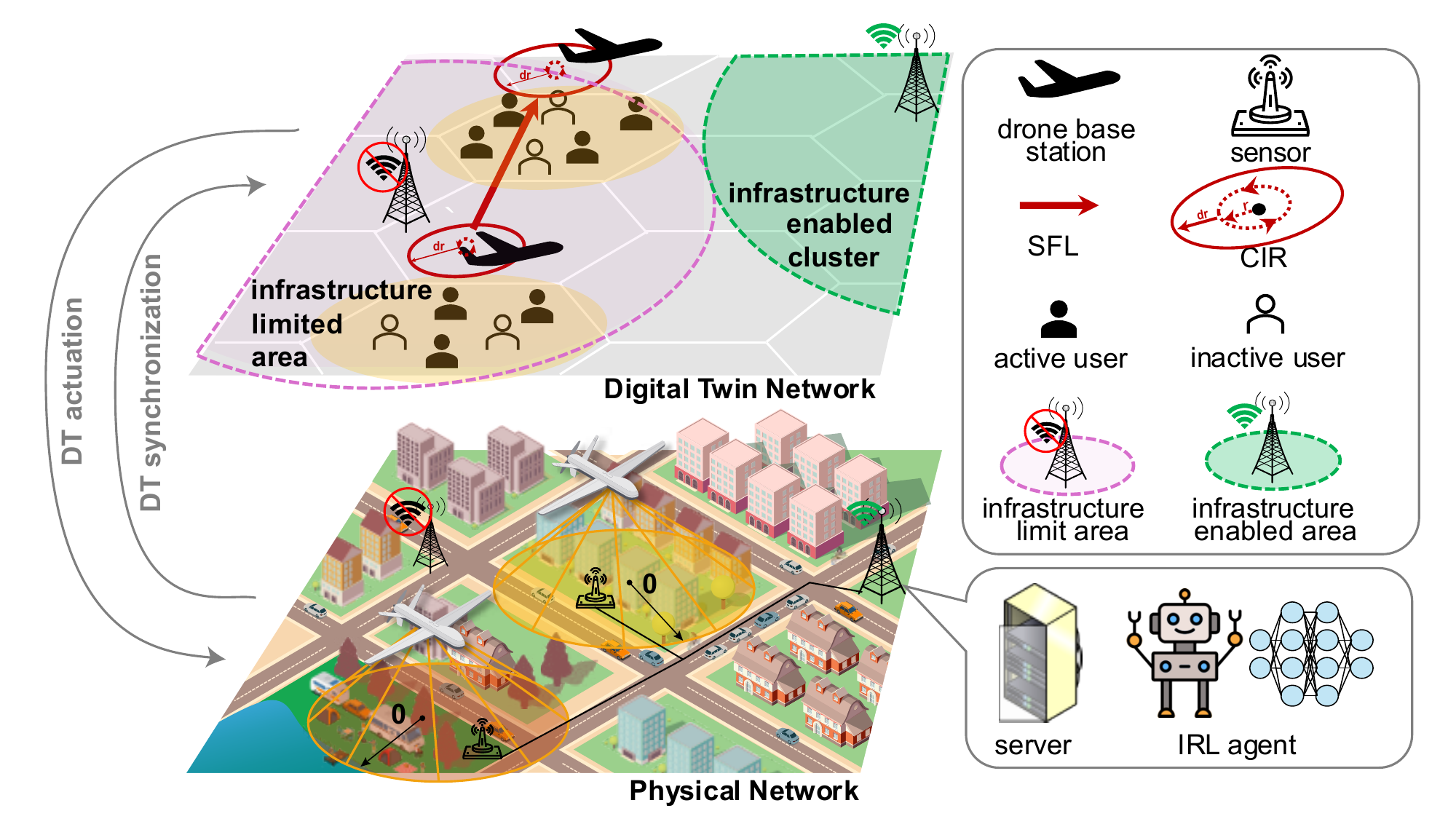}
  \caption{\footnotesize{Network topology.}}
  \label{Fig. 1}
  \centering
  \vspace{-0.5cm}
\end{figure}

\vspace{-0.2cm}
\subsection{Communication Performance Analysis}

In the considered network, the users transmit data over a set of uplink resource blocks (RBs) using the orthogonal frequency division multiple access (OFDMA) technique. 
The achievable data rate at the link between the DBS and user $u$ is given by
\begin{equation}
    c_u = \beta_{\text{LoS}}^{u} B \log_2 \left( 1 + \gamma_{\text{LoS}}^{u} \right) + \beta_{\text{NLoS}}^{u} B \log_2 \left( 1 + \gamma_{\text{NLoS}}^{u} \right),
\end{equation}
within which $\beta_{\text{LoS}}^{u} = \left( 1 + \phi \exp \left( -\frac{\varphi}{180/\pi} \theta_u \right) \right)^{-1}$ represents the probability of a line-of-sight (LoS) data access,
and $\beta_{\text{NLoS}}^{u} = 1 - \beta_{\text{LoS}}^{u}$  represents the probability of a non-line-of-sight (NLoS) data access. Here, $\phi$ and $\varphi$ denote the constant values that depend on the communication environment, while $\theta_u$ is the elevation angle between the DBS and user $u$. $\gamma^{\textrm{LoS}}_{u} =P(N_0 B 10^{\frac{h^{\textrm{LoS}}_{u}}{20}})^{-1}$ and $\gamma^{\textrm{NLoS}}_{u} = P(N_0 B 10^{\frac{h^{\textrm{NLoS}}_{u}}{20}})^{-1}$ represent the signal-to-noise ratio (SNR) at the LoS and NLoS links, respectively. Here, $P$ is the transmission power at user equipments (assumed to be equal for all users), $N_0$ is the noise power spectral density, and $B$ is the RB bandwidth (assumed to be equal for all RBs). %To avoid LoS interference to ground links, we assume that the DBS uses its own dedicated frequency band. 
Path losses $h^{\textrm{LoS}}_{u} = 20 \log \left( \frac{4\pi f_c d_u}{c} \right) + \zeta^{\textrm{LoS}}_{u}$ and $h^{\textrm{NLoS}}_{u} = 20 \log \left( \frac{4\pi f_c d_u}{c} \right) + \zeta^{\textrm{NLoS}}_{u}$ represent, respectively, the path loss in (dB) for LoS and NLoS air-to-ground communication links, with $f_c$ being the carrier frequency of the communication link between the DBS and user $u$. $d_u$ is the distance between user $u$ and the DBS\footnote{Since the DBS flies over a circle with small radius, we assume that $d_u$ remains the same value during the DBS' data collection.}, while $c$ is the speed of light. $\zeta^{\textrm{LoS}}_{u}$ and $\zeta^{\textrm{NLoS}}_{u}$ are the additional path losses at the LoS and NLoS links between the DBS and user $u$, respectively, which are assumed to follow Gaussian distributions with different parameters. The path loss values are assumed to be stable as the distance between the DBS and user $u$ only experiences slow changes when the DBS flies within the service area.

%Indeed, to satisfy all the user requests within the $k_{th}$ cluster, the DBS need to keep hovering over the cluster along a circular path with radius $r$.
The time needed by the DBS to hover over cluster $k$ for serving user requests is 
\begin{equation}\label{eq:3_single}
D^*_{{k}}=
  \begin{cases}
    \mathop {\max }\limits_{u\in \mathcal{U}^*_{{k}}} D_{u} - \frac{2d_r}{V} & \text{if $\mathop {\max }\limits_{u\in \mathcal{U}^*_{{k}}} D_{u} - \frac{2d_r}{V} \geq 0$,} \\
    0 &  \text{otherwise,} \\
  \end{cases}
\end{equation}
in which $D_{u}=\frac{b_u}{c_{u}}$ is the transmission delay of the communication link between user $u$ and the DBS. $\mathop {\max }\limits_{u\in \mathcal{U}^*_{{k}}} D_{u}$ is the time that the DBS used to serve all the users in cluster ${\xi}_{k}$. We define $\mathcal{U}^*_{{k}}=\left\{u\left| {u\in\mathcal{U}_{{k}}, T-\tau^*_{k} \le t_u\le T-\tau_{k}} \right.\right\}$ as the set of active users that can be served by the DBS in cluster ${{\xi}_{k}}$. 
Here, $\tau_{k}=T-\sum_{\kappa=0}^{k-1}\frac{d_{{\kappa},{\kappa+1} }}{V}-\sum_{\kappa=1}^{k}D^*_{{k}}$ is referred to as the \emph{available service time} of the DBS after step $k$, with $d_{\kappa,\kappa+1}$ being the distance between cluster $\xi_{\kappa}$ and cluster $\xi_{\kappa+1}$. It represents the remaining time the DBS is allowed to keep flying after serving cluster $\boldsymbol{\xi}_{k}$.
$\tau^*_{k}=T-\sum_{\kappa=0}^{k-1}\frac{d_{{\kappa},{\kappa+1} }}{V}-\sum_{\kappa=1}^{k-1}D^*_{{\kappa}}$ is referred to as the available service time of the DBS before step $k$.
$\frac{2d_r}{V}$ is the time the DBS consumes during its SLF travel within the service area. 
\vspace{-0.2cm}
\subsection{Flying Velocity of the DBS}
To serve the intended users, the DBS will maintain its flying direction under wind effect \cite{Alez2023VTOLUD}. In such a case, the DBS flying speed will be adjusted as 
\begin{equation}
V = \sqrt{V_g^2 + v_w^2 + 2V_gv_w \cos\rho}. 
\end{equation}
In other words, the DBS flying speed is the vector sum of its ground speed $\vec{v}_g$ and wind speed $\vec{v}_w$. The angle $\rho$ denotes the relative angle between the wind direction and the DBS heading, as shown in Fig. \ref{Fig. UAVWIND}.

 \begin{figure}
  \centering
  \vspace{-0.1cm}
  \includegraphics[width=0.4\textwidth]{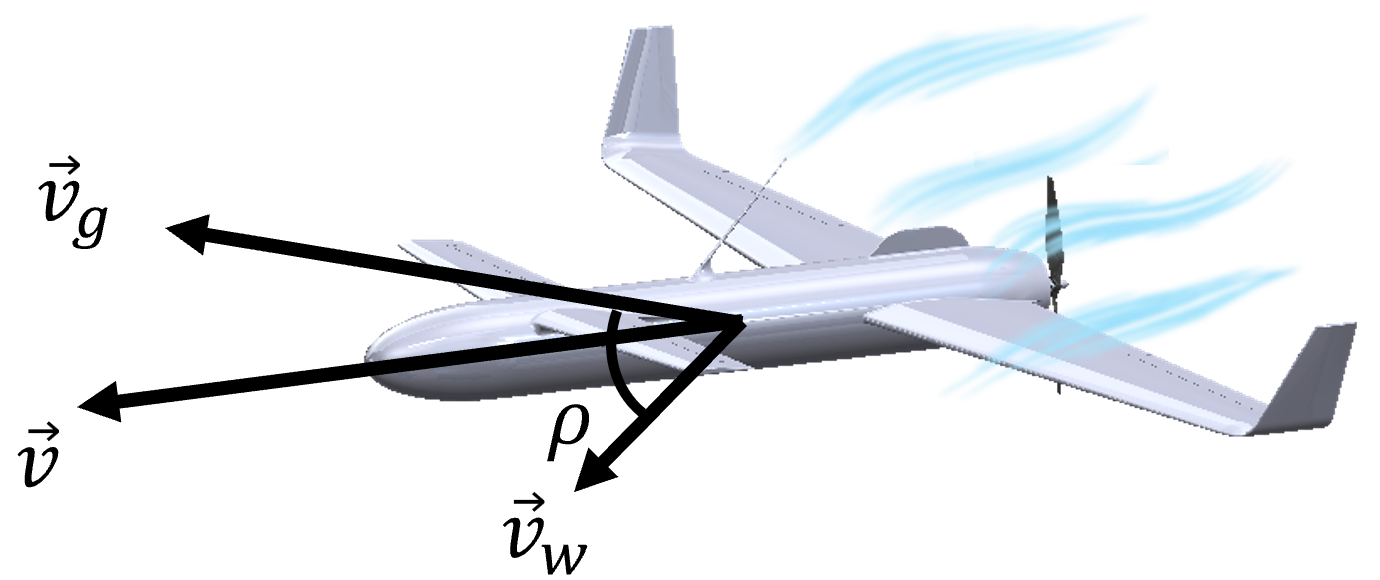}
  \caption{\footnotesize{Reference flying direction used in the mathematical model.}}
  \label{Fig. UAVWIND}
  \centering
  \vspace{-0.5cm}
\end{figure}

\vspace{-0.2cm}
\subsection{Energy Budget of the DBS}
The total energy consumption of the DBS includes two components: communication power and propulsion power. Note that in practice, the communication-related energy is typically negligible compared to the DBS's propulsion energy, e.g., a few watts versus thousands of watts \cite{Thibbotuwawa2019Energy}, therefore, it is omitted in the energy budget of the DBS. The propulsion energy consumed during SLF between cluster ${\xi}_{k}$ to cluster ${\xi}_{k+1}$ is given by
\begin{equation}
    e^{\text{SLF}}_{k,k+1} = \left( \frac{2d_r}{V} + \frac{d_{k, k+1}}{V} \right) \left( c_1 V^3 + \frac{c_2}{V} \right),
\end{equation}
in which $c_1$ and $c_2$ are constants related to the aircraft's characteristics, such as weight and aerodynamics \cite{Alez2023VTOLUD}. $\frac{2d_r}{V} + \frac{d_{k, k+1}}{V} $ is the time needed by the DBS to travel in a SLF from cluster ${\xi}_{k}$ to cluster ${\xi}_{k+1}$. The propulsion energy consumed for the circular flight at cluster ${\xi}_{k}$ is
\begin{equation}
e^{\textrm{CIR}}_{k} = D^*_{{k}} \left[ \left( c_1 + \frac{c_2}{g^2 r^2} \right) V^3 + \frac{c_2}{V} \right],
\end{equation}
in which $g$ is the gravitational acceleration.

\vspace{-0.3cm}
\subsection{Hit rate}
The service status of cluster $\xi_k$ is captured as the portion of served users, which we refer to as the hit rate at cluster $\xi_k$,  and is calculated as
\begin{equation}\label{eq:ssrate_single}
\begin{split}
\mu_{k}\left(\boldsymbol{\xi}\right)=\frac{\sum_{u\in\mathcal{U}}\mathds{1}_{\left\{u\in \mathcal{U}_{{k}}, T-\tau^*_k \le t_u\le T-\tau_{k}\right\}}}{\sum_{u\in\mathcal{U}}\mathds{1}_{\left\{0\le t_u\le T\right\}}},
\end{split}
\end{equation} 
within which $\mathds{1}_{\left\{x\right\}}=1$ when $x$ is true, otherwise, $\mathds{1}_{\left\{x\right\}}=0$. 
Here, $\sum_{u\in\mathcal{U}}\mathds{1}_{\left\{u\in \mathcal{U}_{k}, T-\tau^*_k \le t_u\le T-\tau_{k}\right\}}$ is the number of active users served by the DBS in cluster $\xi_{k}$, while $\sum_{u\in\mathcal{U}}\mathds{1}_{\left\{0\le t_u\le T\right\}}$ is the number of all active users within the studied time duration.

\vspace{-0.3cm}
\subsection{Problem Formulation}
In the considered model, an energy-constrained DBS flies among the distributed user clusters to cover access requests. The utility of this DBS can get by flying along trajectory $\boldsymbol{\xi}$ is defined as 
\begin{equation}\label{eq:utility}
\begin{split}
G\left(\boldsymbol{\xi}\right) = \sum_{k=1}^{K} w_{k} \frac{{\mu}_{k}\left(\boldsymbol{\xi}\right)}{K}.
\end{split}
\end{equation}
In other words, the utility of the DBS is the weighted sum of cluster hit rates. The goal of this DBS is to find an optimal trajectory that maximizes the expected utility with unknown intents, which is formulated as
\begin{subequations}\label{opt}
\begin{align}
    & \max_{\boldsymbol{\pi}_d}  \sum_{\boldsymbol{\xi}\in\mathcal{E}} G\left(\boldsymbol{\xi}\right) \prod_{k=1}^{K} \pi_d\left(\xi_{k+1} \left| \xi_k \right.\right)     &\tag{\ref{opt}}\\
    {\text{s.t.} \quad} & \sum_{\boldsymbol{\xi}\in\mathcal{E}} \prod_{k=1}^{K} \pi_d\left(\xi_{k+1}\left| \xi_k \right.\right) = 1, \label{c1a}\\
    {} & \sum_{\xi \in \mathcal{C} \cup \left\{O\right\}} \pi_d\left(\xi_{k+1}\left| \xi_k\right.\right) = 1, \quad \forall \boldsymbol{\xi} \in \mathcal{E}, k \in \mathcal{K}, \label{c1b}\\
    {} & 0 \leq \pi_d\left(\xi_{k+1}| \xi_k \right) \leq 1, \quad \forall \boldsymbol{\xi} \in \mathcal{E}, k \in \mathcal{K}, \label{c1c}\\
    {}&  e_{\text{k}} \leq E, \label{c1d}
\end{align}
\end{subequations}
within which \( \mathcal{E} \) is the set of all possible trajectories of the DBS. \( \pi_d\left(\xi_{k+1}\left| \xi_k \right.\right) \) is the probability that the DBS moves toward cluster \( \xi_{k+1} \in \mathcal{C} \cup \left\{O\right\} \) after successfully serving cluster \( \xi_k \), and \( \boldsymbol{\pi}_d = \left[\pi_d\left(\xi_{k+1}\left| \xi_k \right.\right)\right]_{k \in \mathcal{K}} \) defines the probability distribution of its flying directions (i.e., next cluster to be served). Constraint (\ref{c1a}) means that the DBS has to fly along viable trajectories from the set \( \mathcal{E} \). (\ref{c1b}) ensures that, at any step \( k \) of its trajectory, the DBS has to choose either serve a cluster \( \mathcal{C} \) or return to the origin. Constraint (\ref{c1c}) ensures that the probabilities are valid. Meanwhile, (\ref{c1d}) means that the energy consumed by the DBS for serving user clusters and returning to the origin,  $e_{k} = \sum_{\kappa=1}^{k} \left( e^{\text{SLF}}_{\kappa-1,\kappa} + e^{\text{CIR}}_{\kappa} \right)+\frac{d_{{\kappa},{O} }}{V}\left( c_1 V^3 + \frac{c_2}{V} \right)$, should not exceed this DBS's energy capacity, and $E$ denotes the maximum available energy of the DBS. From (\ref{opt}), we can see that, the DBS seeks to fly over the trajectory that maximizes the expected utility $G\left(\boldsymbol{\xi}\right) \prod_{k=1}^{K} \pi\left(\xi_{k+1}\left| \xi_k \right.\right)$. 

The expected DBS utility is determined by the DBS service priorities $w_k$, the internal and external PNO properties, such as communication performance, wind velocity, user service requests, flying velocity, and energy budget of the DBS. Thus, once some physical network properties change, the DBS needs to change its maneuvering strategy. In fact, if a physical network property change is not correctly perceived (e.g. due to delayed or inaccurate sensing of the changes), the performance of the DBS will degrade
\begin{theorem}
    Errors on the perception of physical network property, $ \boldsymbol{ \Delta s}_k$, can induce degraded DBS utility.
\end{theorem}
\begin{proof}
    See Appendix A.
\end{proof}
In fact, the performance loss increases exponentially with the errors of physical network property perception as proved in Theorem 2.
\begin{theorem}
     The performance loss of DBS increases with property perception error $||\Delta \boldsymbol{s}_k||$ at step $k$ of the trajectory, in proportion to $e^{-L_{\pi}\epsilon}$. Here, $\epsilon$ is the bound of perception error, i.e., $||\boldsymbol{\Delta s}_k||\le\epsilon$. $L_{\pi}$ is the local Lipschitz constant. 
\end{theorem}
\begin{proof}
    See Appendix B.
\end{proof}
\noindent In other words, the accurate, timely perception of network changes remains critical for appropriate decisions in the drone network. However, considering the fact that the user requests and the wind velocity can change without a trackable pattern, it will be very challenging for the myopic, portable DBS to accurately sense or predict them. 
Since the DT system can access all the property changes (through DT synchronization and prediction), it will be ideal if the DT can directly suggest strategy adjustments to the DBS upon network changes. However, traditional DT systems \cite{9851770,dt-without2} cannot infer the service intent of a DBS, such that they are not applicable to the considered algorithm. To this end, we propose an IRL assisted DT actuation solution that infers DBS intents, and evolves the DBS flying strategy at any physical network changes.

% 9709537

% \vspace{0.3cm}
\section{Proposed Inverse Reinforcement Learning Solution}
We now introduce the IRL assisted DT actuation algorithm that integrates GAN with the advantage actor-critic (A2C) framework\cite{ho2016generativeadversarialimitationlearning}. As a classical RL algorithm, A2C finds the optimal synthetic DBS trajectory by running gradient descent over the strategy space to maximize the expected utility. However, this process assumes full knowledge of the expected utility formulation including the DBS intents, which is not available to the DT system. The GAN model is deployed to generate DBS trajectories that imitate the optimal ones under a wide variety of network conditions, which enables the learning of optimal trajectories across all potential network states without access to DBS intents or intensive environmental exploration. In the following, we first introduce the components of the proposed IRL algorithm. Then, we explain how to use this algorithm to actuate DBS trajectory adjustments.

\vspace{-0.2cm}
\subsection{Components of the Proposed Algorithm}
The proposed IRL algorithm consists of six components:

% In the trajectory optimization of DBSs across various environments, the factors influencing system performance are multifaceted, making it challenging to evaluate them with a single uniform reward function.
% Unlike existing reinforcement learning-based methods for optimizing DBS trajectories, the proposed IRL-based approach does not depend on a predefined reward function. 

\begin{itemize}
    \item Agent: Our agent is the target DBS whose goal is to maximize its expected utility.
    % 物理意义
    \item Actions: The actions of the agent are the clusters or origin it targets at each time step. At step \(k\), the action is denoted as \(a_k \in \mathcal{C} \cup \{O\}\). The vector of actions up to step \(k\) is denoted as \(\mathbf{a}_k = [a_0, a_1, \dots, a_k]\), with \(a_0\) being the origin location, and
    \(a_k\) representing the \(k\)-th location the DBS targets.
    \item States: The state at each step \(k\) includes the physical DBS's internal and external properties, and is denoted as \(\boldsymbol{s}_k = [a_{k-1}, \mu_k(\xi), \boldsymbol{b}, \boldsymbol{t}, e_k, \boldsymbol{\vec{v}_{w}}]\), within which
    \begin{itemize}
        \item \(a_{k-1}\): The DBS's location at the current step, which is represented by the cluster it is currently serving.
        \item \(\mu_k(\xi)\): The hit rate achieved at step \(k\), representing the percentage of successfully served users.
        \item \(\boldsymbol{b}\): The quantities of user requests in bits. 
        \item \(\boldsymbol{t}\): The occurrence time of the user requests.
        \item \(e_k\): The total energy consumption and essential energy required for the DBS to return to the origin $O$ at the $k\text{-th}$ step, which cannot exceed the energy limitation $E$.
        \item \(\boldsymbol{\vec{v}_{w}}\): The wind velocity.

    \end{itemize}

    % \item Expert demonstrations: The expert demonstrations $\boldsymbol{\tau}_{E}$ are observed sequences of actions and states generated by an expert agent in physical network following optimal policy $\pi_{E}$. These demonstrations are pre-collect data in the DT system, enabling the DBS to infer the underlying reward function for guiding the learning of trajectory design.
    \item Reward:
    The reward of the DBS is denoted as $R({a}_k|\boldsymbol{s}_k) = w_k\mu_k(\xi)$, which represents the weighted hit rate the DBS achieves with action $\boldsymbol{a}_k$ at state $\boldsymbol{s}_k$ when following trajectory $\boldsymbol{\xi}$, such that $\sum_{k=1}^K R({a}_k|\boldsymbol{s}_k) = G(\boldsymbol{\xi})$.
    \item Critic:
    The critic network $\Upsilon_{\boldsymbol{\omega}}(\boldsymbol{s}_k)$ is a neural network that aims to predict the expected return $\mathbb{E}{\pi_{\boldsymbol{\theta}}}[\sum_{k=i}^K \gamma^{k-i} R({a}_k|\boldsymbol{s}_k)]$ for a given state $\boldsymbol{s}_k$, where $\boldsymbol{\omega}$
    denotes the parameters of this critic network and $\gamma$ represents a discount factor ranging from 0 to 1.
    \item Policy: The policy of the DBS, denoted as \(\pi_{\boldsymbol{\theta}}({a}_k | \boldsymbol{s}_k)\), defines the probability of choosing action \({a}_k\) given the current state \(\boldsymbol{s}_k\). This policy is approximated by a deep neural network parameterized by \(\boldsymbol{\theta}\), which takes the agent's current state as input and outputs the probabilities for selecting each possible action.

\end{itemize}
More specifically, the DT system will map the physical DBS into the virtual environment by using a policy network \(\pi_\theta(a_k | \boldsymbol{s}_k)\) of the IRL algorithm as its virtual replica. The policy network maps network state $\boldsymbol{s}_k$ that is monitored or predicted by the DT system into the optimal trajectory that maximizes the expected DBS utility. Thus, using this policy function, the DT system can directly actuate trajectory decisions for the DBS upon environmental changes.

% We first reformulate the problem as a Markov decision process (MDP), and then explain the differences between IRL and traditional RL approaches.

\vspace{-0.2cm}
\subsection{Learning Algorithm}
%\vspace{-0.1cm}
In the proposed algorithm, the goal of the DT system is to find the critic network that explains the trajectories of the physical DBS with an expected DBS utility, $\Upsilon{\boldsymbol{\omega}}(\boldsymbol{s}_k)$, and a strategy, $\pi_{\boldsymbol{\theta}}({a}_k | \boldsymbol{s}_k)$, that maximizes this expected utility at any physical network state. Given such goals, the DT system needs to update the critic network $ \Upsilon_{\boldsymbol{\omega}}(\boldsymbol{s}_k)$ toward the minimal difference between the highest expected utility achieved with learned optimal policy $\pi^*_{\boldsymbol{\theta}}(a_k|\boldsymbol{s}_k)$, i.e., $\mathbb{E}{\pi^*_{\boldsymbol{\theta}}}[\Upsilon_{\boldsymbol{\omega}}(\boldsymbol{s}_k)]$, and highest achievable utility $\mathbb{E}{\pi_{E}}[\Upsilon_{\boldsymbol{\omega}}(\boldsymbol{s}_k)]$ using optimal strategy $\pi_{E}$. Thus, the critic parameters are updated with the goal of 

\begin{equation}\small 
\label{eq:update-policy}
\begin{split}
    \min_{\boldsymbol{\omega}} \!\left(\max_{\boldsymbol{\theta}}\! -H(\pi_\theta) \!+\! \mathbb{E}{\pi_{\boldsymbol{\theta}}}[\Upsilon{\boldsymbol{\omega}}(\boldsymbol{s}_k)]\right) \!- \!\mathbb{E}{\pi_E}[\Upsilon_{\boldsymbol{\omega}}(\boldsymbol{s}_k)] \!-\! \psi(\Upsilon_{\boldsymbol{\omega}}),
\end{split}    
\end{equation}
within which $\max_{\boldsymbol{\theta}} -H(\pi_{\boldsymbol{\theta}}) + \mathbb{E}{\pi_{\boldsymbol{\theta}}}[\Upsilon{\boldsymbol{\omega}}(\boldsymbol{s}_k)]$ represents the update of policy $\pi_{\boldsymbol{\theta}}({a}_k | \boldsymbol{s}_k)$ towards the maximal utility. $H(\pi_{\boldsymbol{\theta}}) = \sum_{k=0}^K \gamma^k \mathbb{E}_{\pi_{\boldsymbol{\theta}}}[-\log \pi_{\boldsymbol{\theta}}(\Upsilon{\boldsymbol{\omega}}(\boldsymbol{s}_k))]$ is a discounted entropy that promotes exploration in early policy training steps, and then encourages exploitation upon training convergence. Meanwhile, $\psi(\Upsilon_{\boldsymbol{\omega}})$ is a convex reward regularizer that prevents overfitting during training on finite datasets and ensures the learning algorithm properly recognizes and imitates physical trajectories.

In real world wireless networks, the DBS-environment interaction, especially for the interaction with environment under rare, or unexpected user requests and weather conditions, is very limited. To update the critic and policy networks through such limited observation of physical networks, that is, limited benchmarked optimal DBS trajectories, the DT system leverages a GAN model to generate synthetic trajectories that preserve the policy as the physical DBS. This GAN consists of a generator, which aims to generate trajectories that can be matched to the physical DBS trajectories, and a discriminator, which distinguishes the generated trajectories from the physical trajectories. This way, the GAN model can imitate DBS trajectories without directly consulting its underlying intent, promising a privacy preserved, intent-based control of the DBS. Being integrated into the IRL framework, this GAN will take the DBS policy ($\pi_{\boldsymbol{\theta}}(a_k|\boldsymbol{s}_k)$) as its generator. Given the conflicting goals on making the generated trajectories like-real and locating the generated trajectories, the generator $\pi_{\boldsymbol{\theta}}(a_k|\boldsymbol{s}_k)$ and discriminator $\mathcal{D}_{\eta}$ (a neural network parameterized by $\eta$) are updated by solving the following min-max game:
\begin{equation} \label{eq:update-d}
\begin{split}
    \max_{\boldsymbol{\eta}} \min_{\boldsymbol{\theta}} 
    \mathbb{E}_{\pi_{\boldsymbol{\theta}}}
    [\log(\mathcal{D}_{\boldsymbol{\eta}}(\boldsymbol{s}_k, a_k))] \!+ \!
    \mathbb{E}_{\pi_E}[\log(1\! - \!
    \mathcal{D}_{\boldsymbol{\eta}}(\boldsymbol{s}_k, a_k))].
\end{split}
\end{equation}
Then, the resulting generator $\pi_{\boldsymbol{\theta}}$ can accurately imitate the observed DBS policies, and guide the DBS once similar user requests and weather are witnessed or predicted. Here, the expectation is taken over current policy $\pi_{\boldsymbol{\theta}}$ and historical policy $\pi_{E}$. Specifically, DBS's policy $\pi_{\boldsymbol{\theta}}$ is updated using discriminator's output $\log(D_{\boldsymbol{\eta}}(\boldsymbol{s}_k, a_k))$ as a surrogate reward, encouraging the policy toward regions of the state-action space that resemble the physical DBS behavior. Then, in cases where unexpected user requests are monitored or weather is predicted by the DT, policy $\pi_{\boldsymbol{\theta}}$ can directly guide the DBS as if the DBS had already been trained in similar wireless environments, including the real world ones and synthetic ones that are generated by the GAN model.

\subsection {Implementation Analysis}

\vspace{-0.3 cm}

\begin{algorithm}
\setlength{\abovedisplayskip}{-3 pt}
\setlength{\belowdisplayskip}{-3 pt}
\SetAlgoLined
\SetKwInput{KwInput}{Input}
\SetKwInput{KwOutput}{Output}
\SetKwRepeat{KwRepeat}{repeat}{until}
\KwInput{Policy network parameter $\boldsymbol{\theta}^{(1)}$; discriminator network $\boldsymbol{\eta}^{(1)}$; physical trajectories $\mathcal{O}_{\text{phys}}$}

\For{\normalfont{each episode} $i=1,2,\ldots,N$}{
    Collect physical DBS trajectories in $\mathcal{O}_{\text{phys}}$ with balanced cost of physical network property synchronization and DBS trajectory actuation optimality using our method in \cite{Yu2025JIOT_DTN}\;
    
    Generate synthetic trajectories with policy $\pi_{\boldsymbol{\theta}}(a_k|\boldsymbol{s}_k)$\ and collect min $\mathcal{O}_{\text{sim}}$\;
    
    Update discriminator $\boldsymbol{\eta}^{(i)}$ through gradient ascent based on (\ref{eq:update-d})\;
    
    Update policy $\boldsymbol{\theta}^{(i)}$ based on (\ref{eq:update-d})\ to minimize difference between generated trajectories and physical trajectories.
}

\Return Optimized policy network $\pi_{\boldsymbol{\theta}}^{*}$
\caption{The Training Process of the IRL based DT actuation Algorithm}
\label{alg:irl}
\end{algorithm}

\vspace{-0.2 cm}
The training workflow is outlined in Algorithm \ref{alg:irl}. At the beginning of the algorithm, the DT system randomly initializes its policy network parameter $\boldsymbol{\theta}^{(1)}$ and discriminator network parameters $\boldsymbol{\eta}^{(1)}$. The DT then collects $N$ observations of DBS trajectories in a set $\mathcal{O}_{\text{phys}}$, generates synthetic DBS trajectories with policy $\pi_{\boldsymbol{\theta}}(a_k|\boldsymbol{s}_k)$, and stores them in set $\mathcal{O}_{\text{sim}}$. The discriminator $\mathcal{D}_{\eta}$ will be updated through gradient descent to maximize the ability to distinguish between trajectories from set $\mathcal{O}_{\text{phys}}$ and set $\mathcal{O}_{\text{sim}}$. Policy $\pi_{\boldsymbol{\theta}}(a_k|\boldsymbol{s}_k)$ will be updated through policy gradient to minimize the difference between its generated trajectories and the physical trajectories. In essence, this training procedure is repeated as the DT server simulates the physical environment, generates synthetic DBS trajectories, and updates both discriminator parameters $\boldsymbol{\eta}$ and policy parameters $\boldsymbol{\theta}$ to faithfully replicate the physical DBS. Upon the convergence of this IRL algorithm, the DT system will take policy network $\pi_{\boldsymbol{\theta}}(a_k|\boldsymbol{s}_k)$ as the virtual replication of the DBS. Then, once it detects a change in the network state, the DT system uses $\pi_{\boldsymbol{\theta}}(a_k|\boldsymbol{s}_k)$ to adjust DBS trajectories, while the actuation of these adjustments through distributed BSs constitutes the only communication overhead. Meanwhile, the training process is implemented over DT simulation, meaning the update of the IRL model is performed separately from the physical network operations, which reduces the training overhead in terms of time, energy, and DBS hardware.

We then analyze the implementation costs of the IRL algorithm training process. First, the computational complexity of the proposed algorithm is $\mathcal{O}(NK(C_{\pi}+C_{\eta}))$. Here, $N$ represents the number of episodes required for convergence, and $K$ captures the number of synthetic decision steps that are generated and processed in each episode. $C_{\pi}$ defines the computational cost for one step of gradient descent in (\ref{eq:update-d}) for the policy update, while $C_{\eta}$ defines the computational cost for one step of gradient ascent in (\ref{eq:update-d}) for the discriminator update. Notice that, the computation is offloaded from the DBS to the DT server, such that it is more commendable for the resource constrained drone networks compared to traditional onboard computation methods \cite{Zhu2023UAVAoI,meng-2023-multi-uav,guan-2021-uav-trajectory}. Second, the communication overhead of the proposed method includes only the selected maneuvering decisions. These overheads are much less than those of traditional DBS mobility control methods \cite{Liu2024TWC_ISAC,Huang2020UAV, Pervej2025MILCOM} that require the transmission of high dimensional situational awareness or policy model updates among distributed sensors, BSs, and DBSs. Finally, the DBS trajectory actuation performance of the proposed IRL-inspired DT is highly correlated with the number of collected physical DBS trajectories $\mathcal{O}_{phys}$. Specifically, the DT actuation performance increases with the number of collected physical DBS trajectories at the scale of $\mathcal{O}(N_p^{-1/2})$ \cite{Liu2022ICML_GAIL}, with $N_p$ being the number of observed physical trajectories. In this paper, we will leverage our previous work on DT synchronization \cite{Yu2025JIOT_DTN} to collect physical DBS trajectories to balance the communication cost and actuation performance of the DT drone network.

The actuation workflow is outlined in Algorithm \ref{alg:actuation}. At the beginning of the actuation, the DT system deploys optimized policy network $\pi_{\boldsymbol{\theta}}^{*}$. The DT then keeps monitoring and predicting physical current network state $\boldsymbol{s}_k$ and computes action $a_k^*$ based on optimal policy $\pi_{\boldsymbol{\theta}}^{*}$, guiding the physical DBS to execute action $a_k^*$. This process is executed until DBS chooses to return origin $O$ or the energy capacity is reached, upon which the maneuvering of DBS ends. In essence, this actuation process is performed as the DT system senses physical network states, predicts network changes, infers DBS intents, and actuates DBS trajectory adjustments.

Using the proposed algorithm, the time for the DBS performance to recover from environmental changes (e.g., unexpected user demands or wind velocity) includes only the delay of decision transmission between the DT server and the DBS. This is much less than the ones of conventional DBS mobility control methods \cite{zhang-2019-uav-outage-icc,zeng-2017-uav-trajectory} that incur additional delays in sensing and onboard computation. As shown in Fig. \ref{Fig. algorithm12}, the actuation process is implemented with trained policy $\pi_{\theta}^*$. In summary, the online actuation only involves lightweight inference and communication overheads.

 \begin{figure}
  \centering
  %\vspace{+0.3cm}
  \includegraphics[width=0.47\textwidth]{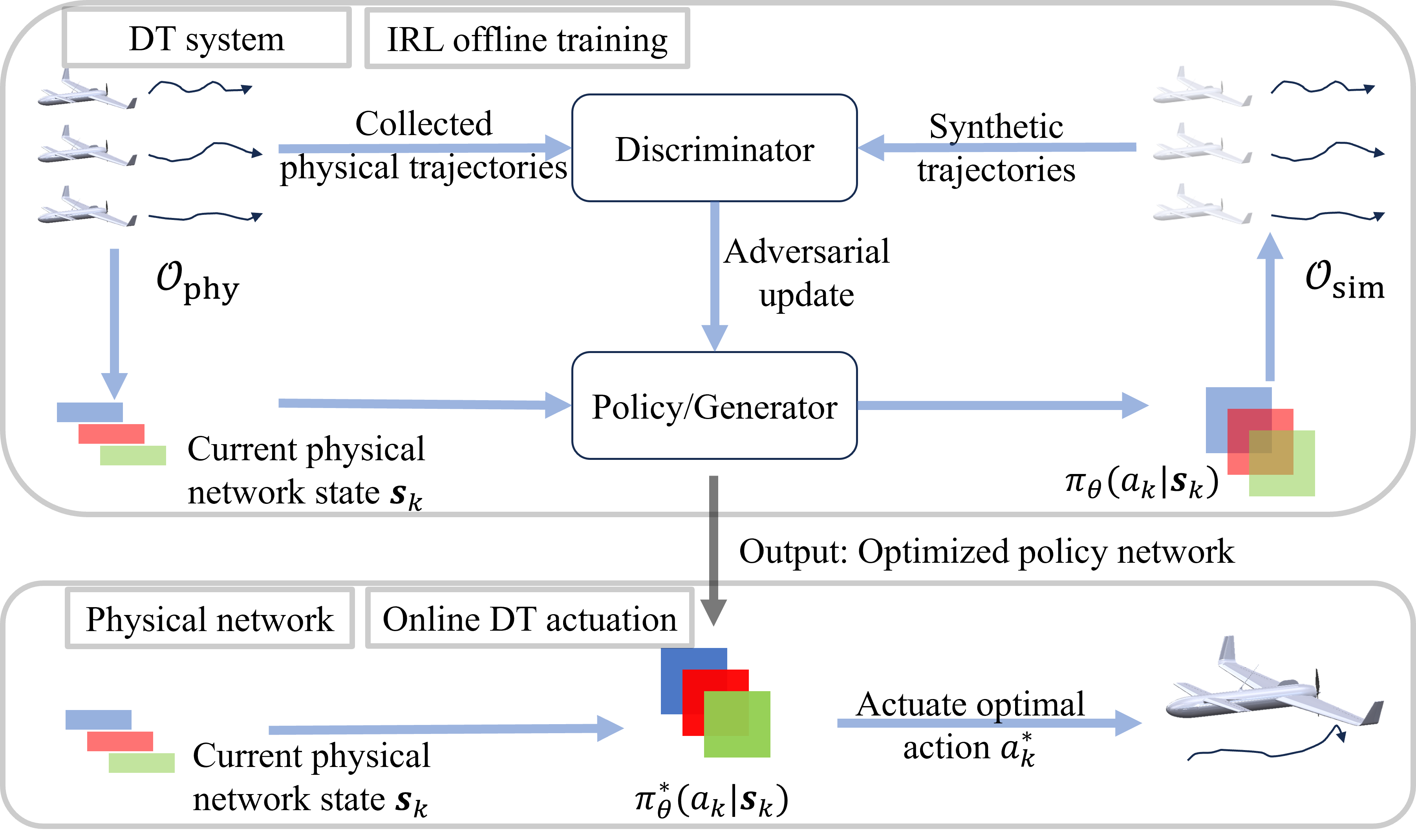}
  \caption{\footnotesize{Schematic illustration of the training and actuation of drone network DT.}}
  \label{Fig. algorithm12}
  \centering
  \vspace{-0.5cm}
\end{figure}

\begin{algorithm}
\setlength{\abovedisplayskip}{-3 pt}
\setlength{\belowdisplayskip}{-3 pt}
\SetAlgoLined
\SetKwInput{KwInput}{Input}
\SetKwInput{KwOutput}{Output}
\SetKwRepeat{KwRepeat}{repeat}{until}
\KwInput{Optimized policy network $\pi_{\boldsymbol{\theta}}^{*}$}

\While{$a_k \neq O$ \normalfont{and} $e_k\le E$}{
    Monitor and predict current physic network state $\boldsymbol{s}_k$\;

    Generate optimized action $a_k^*$ using policy network $\pi_{\boldsymbol{\theta}}^{*}$ and actuate it to the physical DBS.

}

\Return Optimal DBS trajectory $\boldsymbol{\xi}^*$
\caption{Online Actuation of Intent-based DBS Trajectory Adjustment}
\label{alg:actuation}
\end{algorithm}

\vspace{-0.5cm}
\section{Simulation Results and Analysis}

\begin{table}[htbp]
\centering
\caption{Environment and Communication Parameters}
\label{tab:env_comm_params}
\vspace{-0.2cm}
\begin{tabular}{|>{\centering\arraybackslash}p{0.16\columnwidth}|>{\centering\arraybackslash}p{0.24\columnwidth}|>{\centering\arraybackslash}p{0.16\columnwidth}|>{\centering\arraybackslash}p{0.24\columnwidth}|}
\hline
\textbf{Parameter} & \textbf{Value} & \textbf{Parameter} & \textbf{Value} \\
\hline
$H$ & 50 m & $V$ & 20 m/s \\
\hline
$E$ & 100 Wh & $b_u$ & 10 Mb \\
\hline
$B$ & 1 MHz & $f_c$ & 2 GHz \\
\hline
$P$ & 23 dBm & $N_0$ & $-174$ dBm/Hz \\
\hline
\end{tabular}
\vspace{-0.2cm}
\end{table}

\begin{comment}
    \begin{table}[htbp]
\centering
\caption{Training Hyperparameters}
\label{tab:train_hyperparams}
\vspace{-0.2cm}
\begin{tabular}{|>{\centering\arraybackslash}p{0.62\columnwidth}|>{\centering\arraybackslash}p{0.22\columnwidth}|}
\hline
\textbf{Parameter} & \textbf{Value} \\
\hline
Actor--critic learning rate & $3\times10^{-4}$ \\
\hline
Discriminator learning rate & $1\times10^{-4}$ \\
\hline
Discount factor $\gamma$ & 0.99 \\
\hline
GAE coefficient $\lambda$ & 0.95 \\
\hline
PPO clipping parameter & 0.2 \\
\hline
Value loss coefficient & 0.5 \\
\hline
Entropy coefficient & 0.05 \\
\hline
Gradient clipping norm & 0.5 \\
\hline
Mini-batch size & 128 \\
\hline
Expert mini-batch size & 512 \\
\hline
Training epochs & 300 \\
\hline
Environment steps per epoch & 4000 \\
\hline
Discriminator updates per epoch & 30 \\
\hline
PPO update repetitions & 10 \\
\hline
Replay buffer size & 20000 \\
\hline
Policy hidden layers & [128, 64] \\
\hline
Discriminator hidden layers & [256, 256] \\
\hline
Dropout rate & 0.3 \\
\hline
\end{tabular}
\end{table}
\end{comment}

\vspace{-0.2cm}

We consider a scenario in which a single DBS serves a 20 km by 20 km geographic area with a battery capacity that supports at most 100 watt-hours. Specifically, we assume that there are 52 users randomly distributed across 5 clusters. 
At every independent run of the Monte Carlo experiment studied in this section, we repeatedly deploy users at some random locations that are uniformly distributed in every cluster, and sample a user access request realization from assumed distributions. The related environment and communication parameters used in this experiment are shown in Table \ref{tab:env_comm_params}.

The policy and critic networks of the proposed algorithm are implemented as two fully connected feedforward neural networks (FNNs), each with two hidden layers of 128 and 64 neurons, respectively. In addition, the discriminator is implemented as a fully FNNs with two hidden layers of 256 neurons each. Both networks are trained with Adam optimizers, where the actor-critic learning rate and the discriminator learning rate are set to $3\times10^{-4}$ and $1\times10^{-4}$, respectively. The discount factor and the generalized advantage estimation coefficient are set as $\gamma=0.99$ and $0.95$, respectively, while the dropout rate of the discriminator is $0.3$.

\begin{figure}[t]
  \centering
  %\vspace{-0.3cm}
  \includegraphics[width=0.43\textwidth]{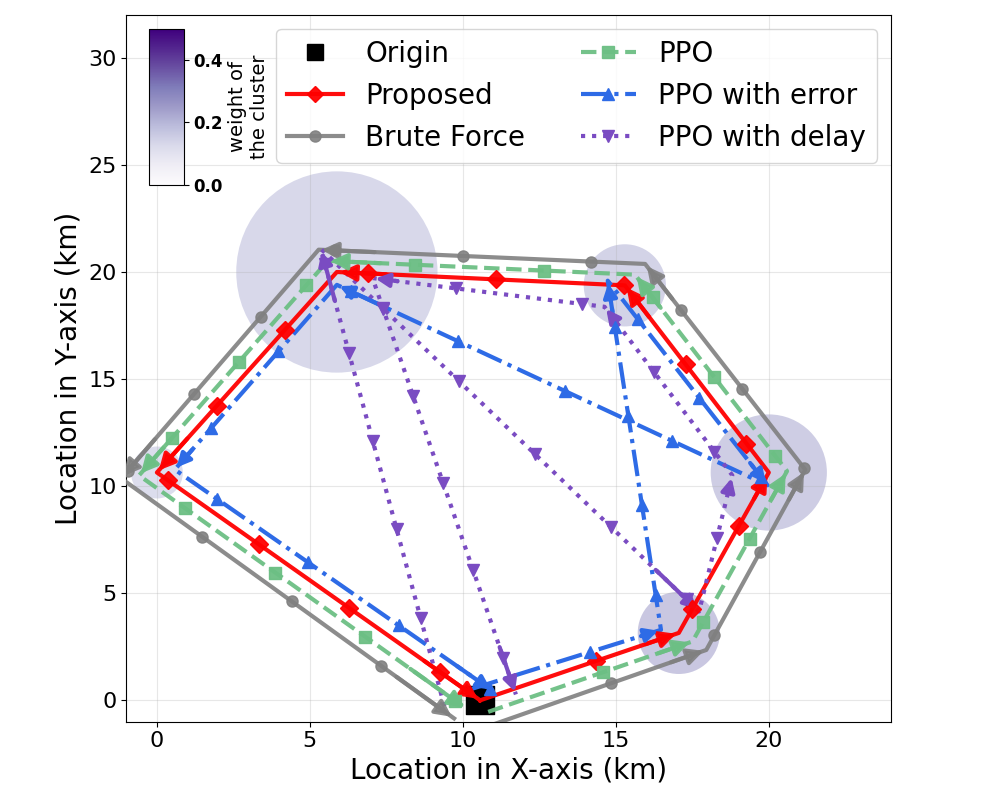}
  \caption{\footnotesize{Snapshot of trajectories resulting from all considered algorithms.}}
  \label{fig:trajectory_comparison}
  \centering
  \vspace{-0.5cm}
\end{figure}

The results of the proposed IRL algorithm are compared with 
\begin{itemize}
    \item Brute force, which assumes perfect knowledge of environmental changes including wind variations and user requests at DBS, finds the optimal DBS trajectories through an exhaustive search over all feasible DBS trajectories. The brute force results serve as a benchmark to verify the optimality of all considered algorithms in this section.
    \item Proximal policy optimization (PPO), which is an RL algorithm that searches for the optimal DBS policy through an on-policy training process. The PPO algorithm is trained and implemented over the DBS with full knowledge of DBS intents, and full exploration of environmental changes with corresponding optimal DBS trajectories. Thus, the comparison between the proposed algorithm and the PPO algorithm can justify the performance of the proposed algorithm in inferring DBS intents and actuating appropriate decisions upon environmental changes.
    \item PPO with erroneous perception, with which DBS makes decisions using the PPO algorithm derived strategies but over erroneous perception of the physical network changes (which are captured by PNO properties such as wind velocity, user requests). By comparing the proposed method with it, we can justify the DBS's performance gain from the usage of accurate modeling of PNO properties in DT.
    \item PPO with delayed perception, with which the DBS will actuate the PPO policy until it receives updated observations of PNO properties. Thus, the comparison between the proposed algorithm and the PPO with delayed perception can verify the DBS's performance gain from harvesting the timely PNO property prediction in DT.
\end{itemize}

\vspace{-0.2cm}
\subsection {Performance Visualization}

We first evaluate the performance of the proposed algorithm under one realization of user requests and wind impacts. In particular, Fig. \ref{fig:trajectory_comparison} shows a snapshot of the DBS trajectories derived from all considered algorithms. The user clusters are represented by purple circles with different color intensities, which correspond to their priorities for being served by the DBS. From the simulation results in Fig. \ref{fig:trajectory_comparison}, we can see that both the PPO algorithm and the proposed algorithm are able to find the optimal trajectories benchmarked by brute force. However, the PPO with erroneous perception and the PPO with delayed perception cannot find the optimal trajectories to cover all user requests. This stems from the fact that erroneous or delayed perception causes PPO to make decisions based on inaccurate or outdated PNO properties, which leads to deviations from the optimal DBS trajectory. This is consistent with our analytical results.

\begin{figure}[t]
  \centering
  \includegraphics[width=0.43\textwidth]{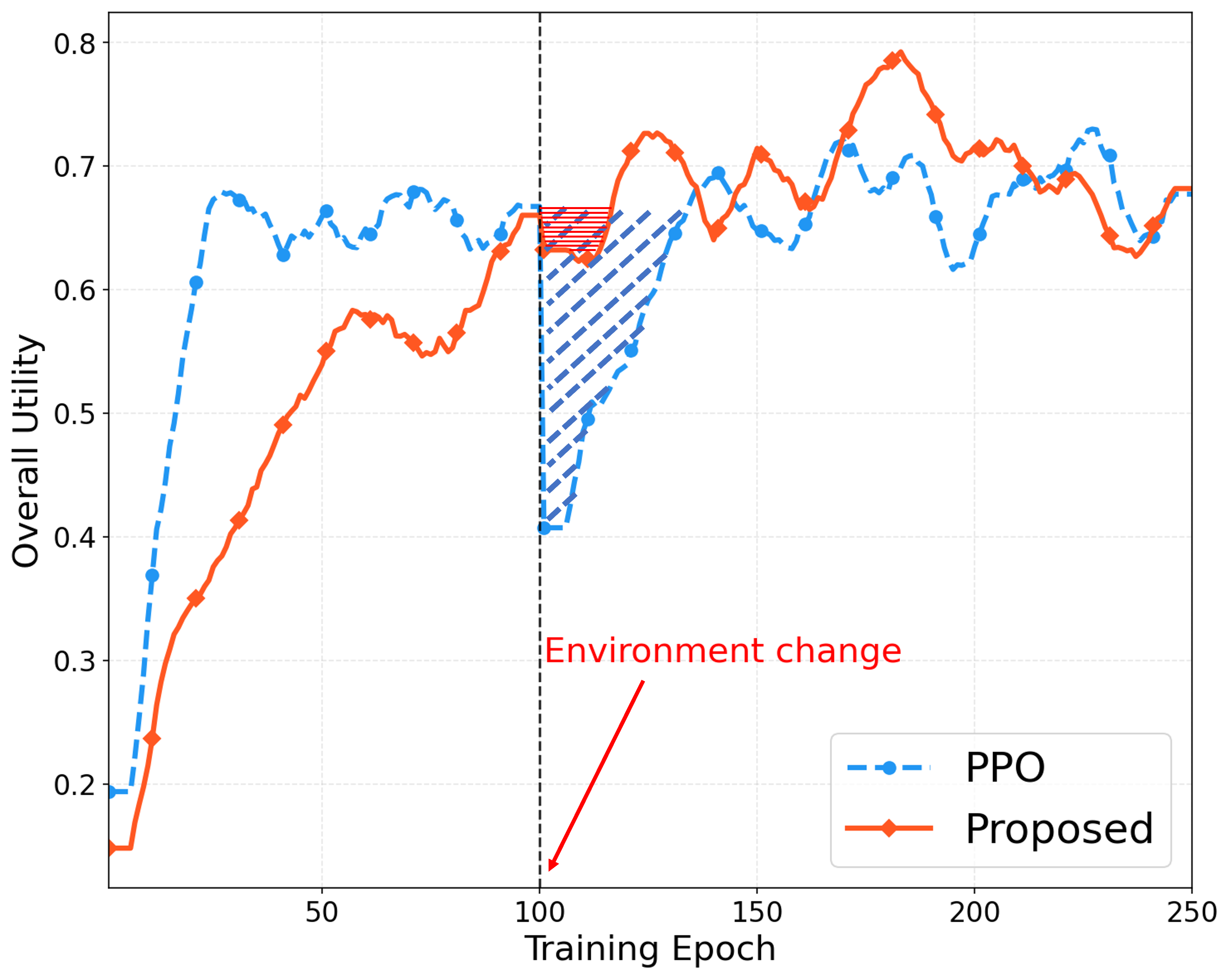}
  \caption{\footnotesize{Training performance under environment shift.}}
  \label{fig:Environment Shift}
  \centering
  \vspace{-0.5cm}
\end{figure}

Fig. \ref{fig:Environment Shift} shows the performance of the PPO policy and the proposed algorithm policy across an environment change (at approximately the 100-th epoch). From Fig. \ref{fig:Environment Shift}, we can see that the DBS performance drops up to 40\% with the PPO algorithm, while it drops only 6\% with the proposed algorithm, upon the environmental change. Meanwhile, it takes approximately 30 steps for the PPO policy, while only about 15 steps for the proposed algorithm policy, to recover DBS performance. This is because the PPO policy and the proposed algorithm policy are fitted to different learning targets. In particular, the PPO policy learns a state-action policy, whereas the policy of the proposed algorithm learns expert state-action distributions over real-world and synthetic wireless environments (defined by the PNO properties). In other words, the proposed IRL algorithm is built to be generalizable to environmental changes. Overall, as shown by the shaded area in Fig. \ref{fig:Environment Shift}, the proposed algorithm yields up to 85\% less DBS performance loss across environmental changes, in comparison to that of the PPO algorithm. This stems from the fact that the proposed algorithm uses the GAN model to synthesize numerous PNO properties, derives optimal DBS trajectories, and allows the direct imitation of such DBS trajectories upon environmental changes.

\vspace{-0.1cm}
\subsection{Evaluation of Performance}

\begin{figure}[t]
  \centering
  \includegraphics[width=0.44\textwidth]{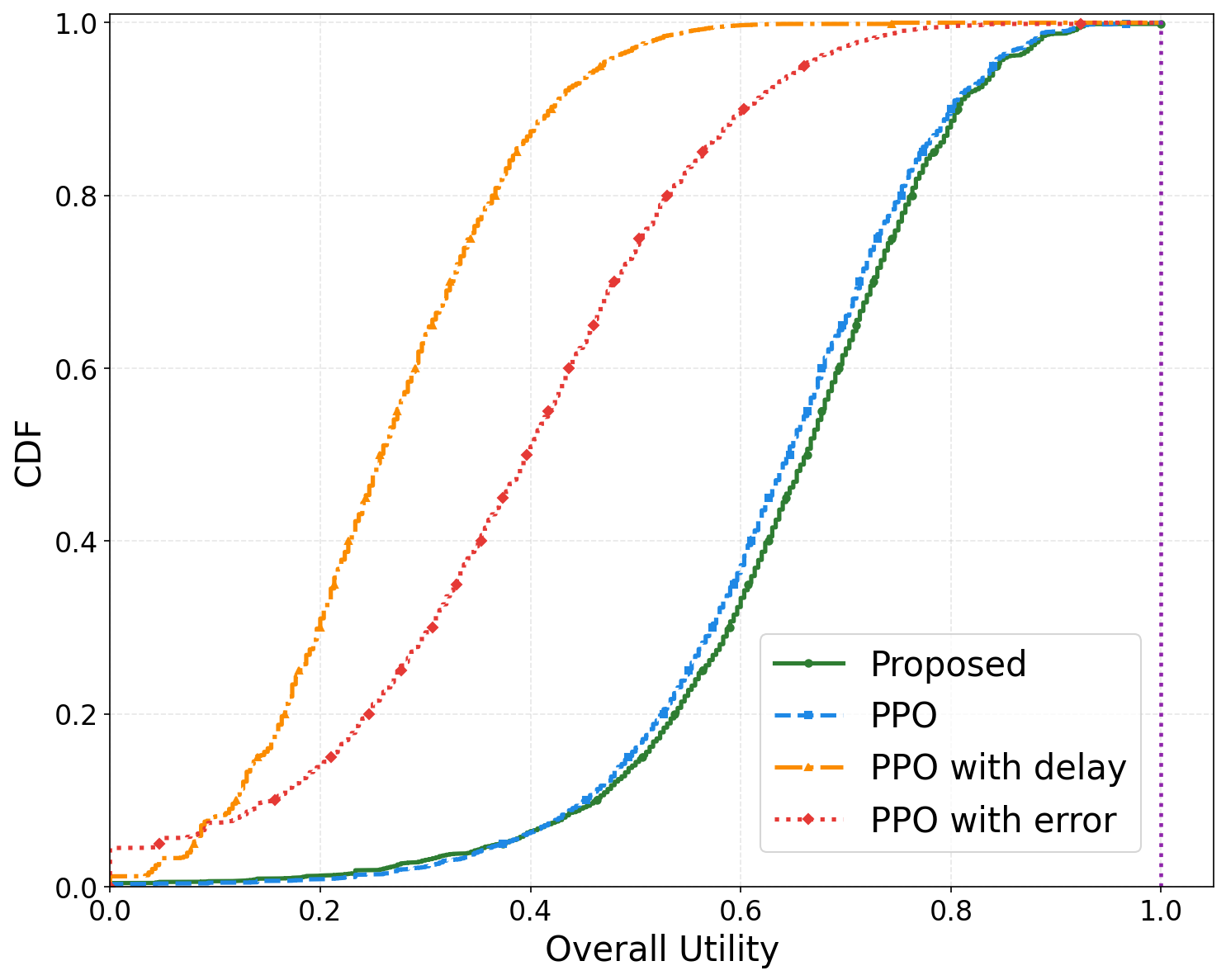}
  \caption{\footnotesize{DBS utility CDF from all considered algorithms.}}
  \label{fig:hitrate_cdf}
  \centering
  \vspace{-0.5cm}
\end{figure}

Next, we evaluate how the proposed generative IRL algorithm achieves near-optimal performance without direct access to DBS intents. Fig. \ref{fig:hitrate_cdf} shows the cumulative distribution function (CDF) of the DBS utility achieved by all considered algorithms. From Fig. \ref{fig:hitrate_cdf}, it can be observed that the CDF of the DBS utility from the PPO algorithm and the proposed algorithm stay closer to that of brute force, compared to the PPO with erroneous perception and the PPO with delayed perception. Particularly, using the proposed algorithm, the probability that the DBS gets a higher than 0.8 utility is about 84\%, which is about 80\% and 79\% higher than that of PPO with delayed perception and PPO with erroneous perception, respectively. This stems from the fact that the proposed algorithm can compensate for perception delays and errors by modeling and predicting PNO properties with a DT. Moreover, this probability is slightly higher than that of the PPO algorithm. This stems from the fact that the proposed algorithm aims to capture the underlying distribution of optimal DBS trajectories in various network conditions, which leads to more generalizable decision-making compared with the PPO algorithm. In summary, the proposed algorithm supports near-optimal DBS utility, as the DT framework compensates perception errors and delays, while the generative IRL model captures the underlying relationship between PNO properties and optimal trajectories to maintain high DBS service coverage.

\begin{figure}[t]
  \centering
  \includegraphics[width=0.44\textwidth]{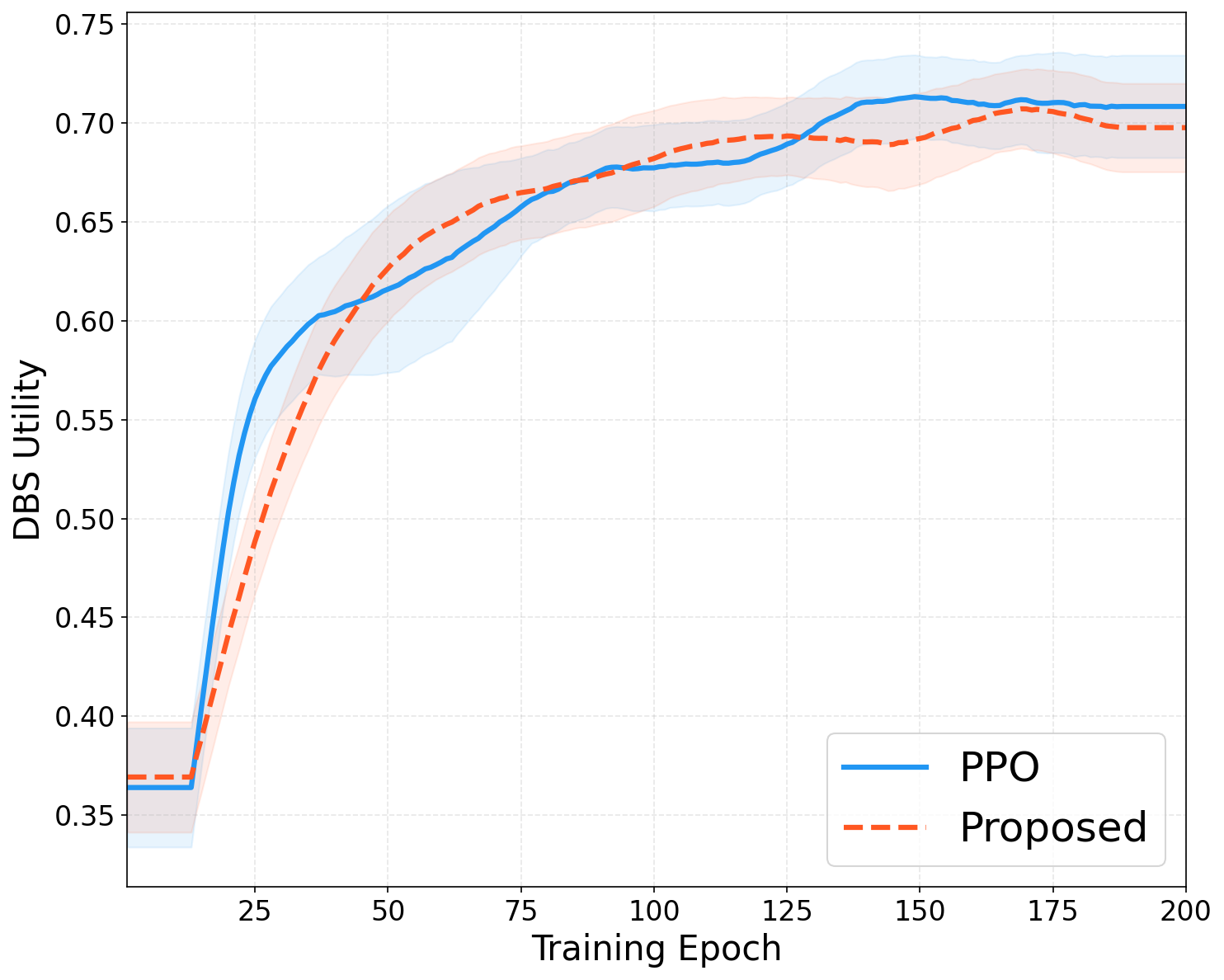}
  \caption{\footnotesize{Convergence of the proposed algorithms and PPO.}}
  \label{fig:mc_convergence}
  \centering
  \vspace{-0.5cm}
\end{figure}

Fig. \ref{fig:mc_convergence} shows the convergence of the PPO algorithm and the proposed algorithm. It can be observed that the proposed algorithm requires approximately 175 iterations to reach convergence, which is 17\% slower than that of the PPO. This is because the proposed algorithm involves an additional discriminator network and an adversarial training process, which introduces additional computation complexity compared with the direct actor-critic optimization in the PPO algorithm. However, compared to the PPO algorithm, the proposed algorithm is trained in a synthetic DT environment, such that this training complexity is commendable in the resource constrained drone networks. Meanwhile, the proposed algorithm provides an average DBS utility of 0.69, which is only 2\% lower than the one of the PPO algorithm. In other words, the generative, adversarial learning process of the proposed algorithm can achieve near-optimal DBS performance despite the lack of access to DBS intents. This stems from the fact that the proposed algorithm uses a GAN model to imitate the optimal intent-based trajectories. Specifically, the discriminator quantifies the discrepancy between generated and optimal trajectories, guiding the generator to refine its policy until the generated trajectories align with the optimal ones. In summary, the proposed algorithm can actuate near-optimal trajectories by offloading the training process into the DT system, which promises adaptive, high performance trajectory design for the DBS in dynamic environments.

\vspace{-0.1cm}
\subsection{Evaluation of Implementation Overheads}

\begin{figure}[t]
  \centering
  \includegraphics[width=0.42\textwidth]{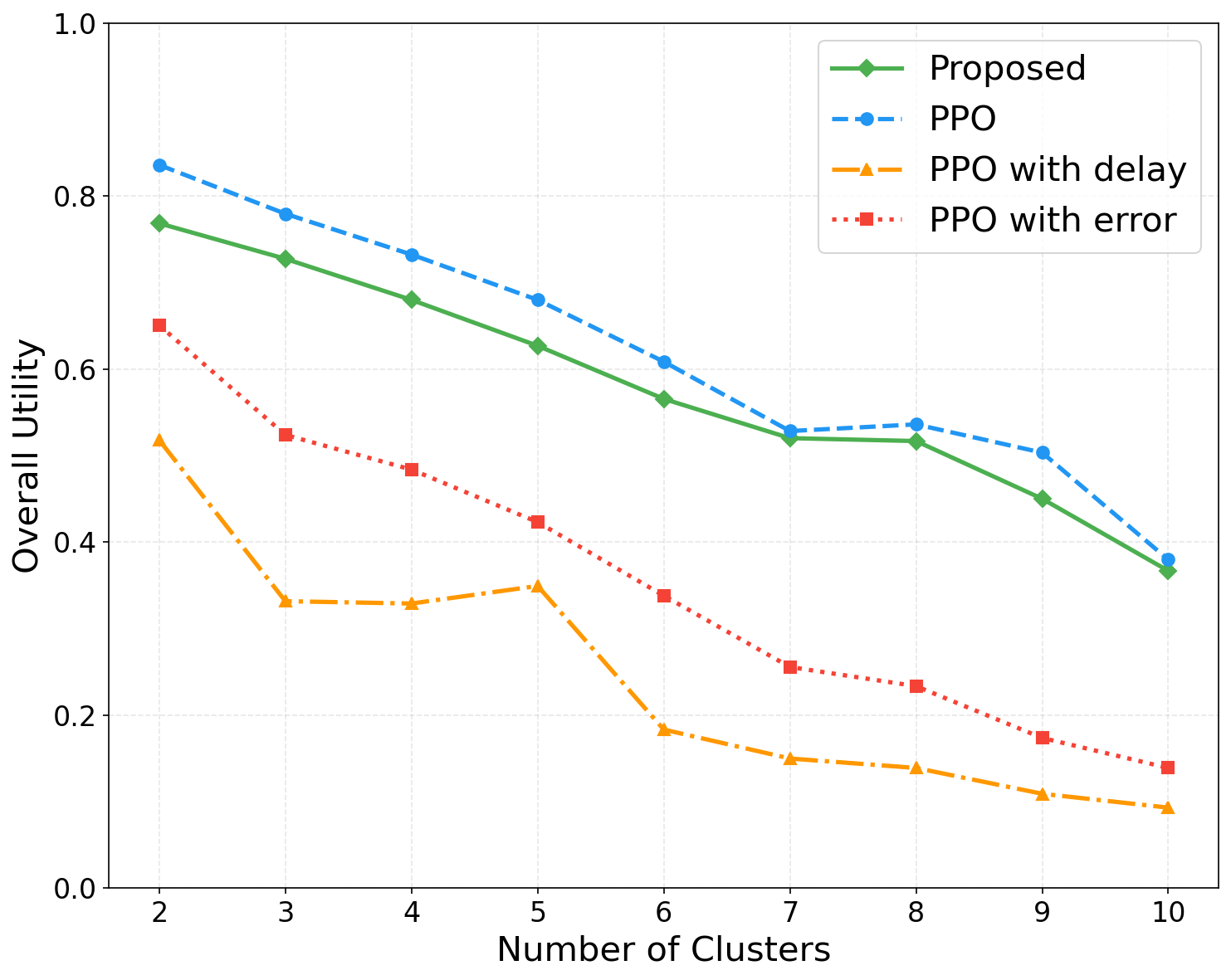}
  \caption{\footnotesize{Sensitivity of all considered algorithms.}}
  \label{fig:scalability_utility}
  \centering
  \vspace{-0.5cm}
\end{figure}

We now evaluate the implementation of the proposed algorithm. Particularly, Fig. \ref{fig:scalability_utility} verifies the sensitivity of all considered algorithms in response to the number of clusters that varies from 2 to 10. From Fig. \ref{fig:scalability_utility}, we observe that the DBS utility from all considered algorithms decreases with the number of clusters. This is because, as the number of clusters increases, user requests become more spatially dispersed and heterogeneous, so not all of them can be served by a single resource-constrained DBS. From Fig. \ref{fig:scalability_utility}, we also see that the proposed algorithm can maintain the DBS utility no less than 90\% of the one of the PPO algorithm. This performance gap stems from the fact that the proposed algorithm must derive the DBS trajectory decisions without access to DBS intents, while the PPO algorithm directly optimizes the DBS utility. Nevertheless, the proposed algorithm maintains a stable utility trend, showcasing its scalability in large sized wireless networks. Meanwhile, the proposed algorithm yields 137\% and 62\% higher average DBS utility, compared to the PPO with delayed perception and the PPO with erroneous perception, respectively, as its DT framework does not have delayed or inaccurate PNO information. It also reduces the performance degradation by 6\% and 22\%, compared to the PPO with delayed perception and the PPO with erroneous perception, respectively. This is because the DT can model and predict environmental changes, which eliminates the impact of delayed and erroneous perception across different numbers of user clusters. Overall, the proposed algorithm can actuate near-optimal decisions to the DBS in various network sizes, and remain scalable to the high-dimensional optimization.

\begin{figure}[t]
  \centering
  \includegraphics[width=0.44\textwidth]{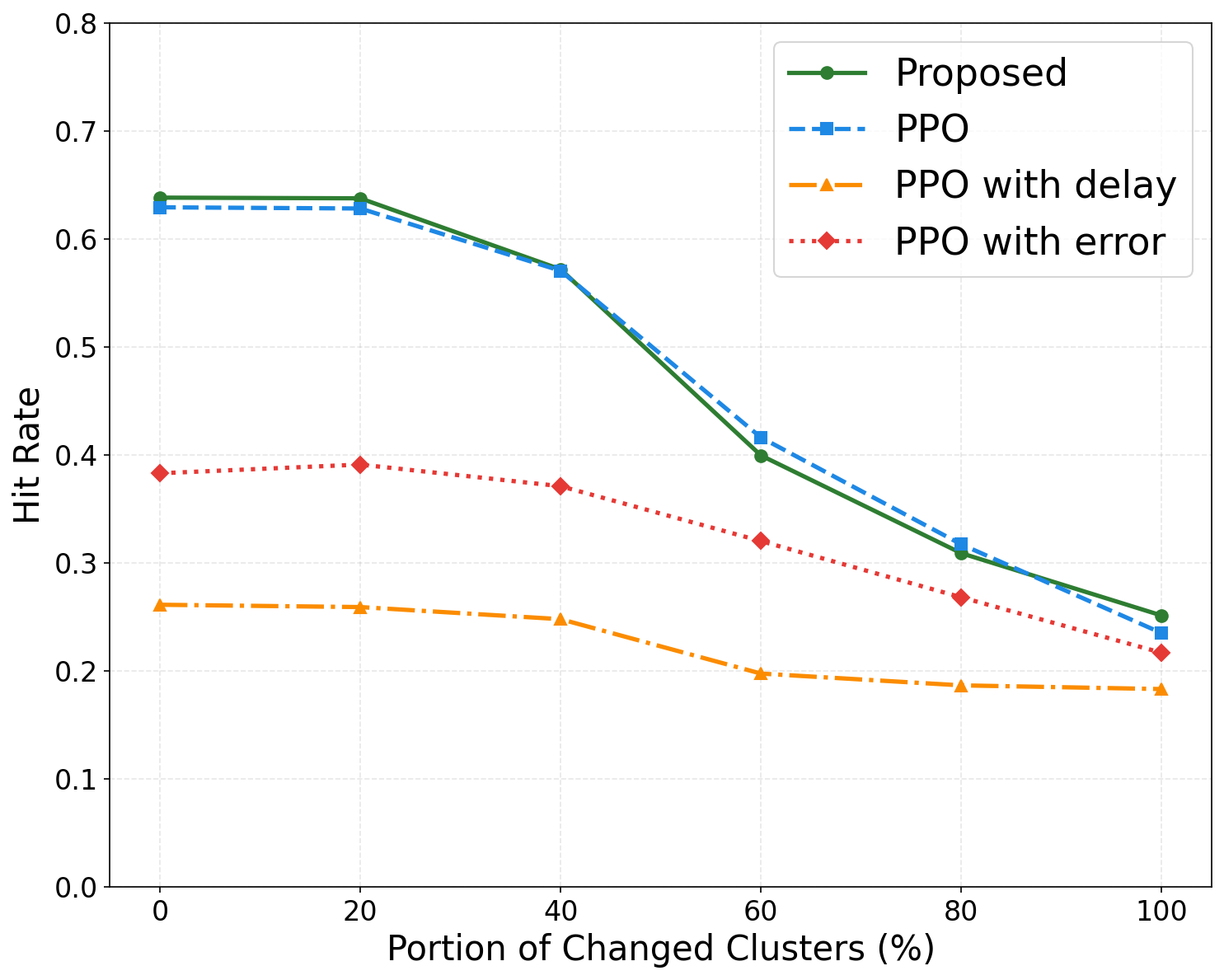}
  \caption{\footnotesize{Reliability of all considered algorithms.}}
  \label{fig:portion_hitrate}
  \centering
  \vspace{-0.6cm}
\end{figure}

Fig. \ref{fig:portion_hitrate} verifies the reliability of all considered algorithms as the properties of PNOs (e.g., user request arrival times) change. In Fig. \ref{fig:portion_hitrate}, we observe that as the portion of clusters with changed properties increases, the performance of all considered algorithms decreases. This stems from the fact that a larger proportion of changed clusters enlarges the mismatch between training and testing environments. Meanwhile, Fig. \ref{fig:portion_hitrate} also shows that the proposed algorithm maintains a hit rate comparable to the one of the PPO algorithm under environmental changes. Specifically, when all clusters have changed properties, the hit rate of the proposed algorithm decreases from about 0.64 to 0.25, corresponding to a degradation of about 61\%, which is 3\% lower than that of the PPO algorithm. This is because the DT framework can anticipate and model the shifted PNO properties, allowing the generative IRL model to adapt the DBS trajectory design to the updated environmental dynamics. We also observe that the PPO with delayed environmental perception yields the lowest hit rate, which is about 28\% lower than that of the proposed algorithm. This performance visualizes how the DBS benefits from the prediction of environmental changes within the DT system. In short, Fig. \ref{fig:portion_hitrate} verifies that, through the modeling and predicting of environmental changes, as well as the generative imitation of intent-based DBS operation in the DT system, the proposed algorithm can maintain reliable performance under environmental changes.

\begin{figure}[t]
  \centering
  \includegraphics[width=0.44\textwidth]{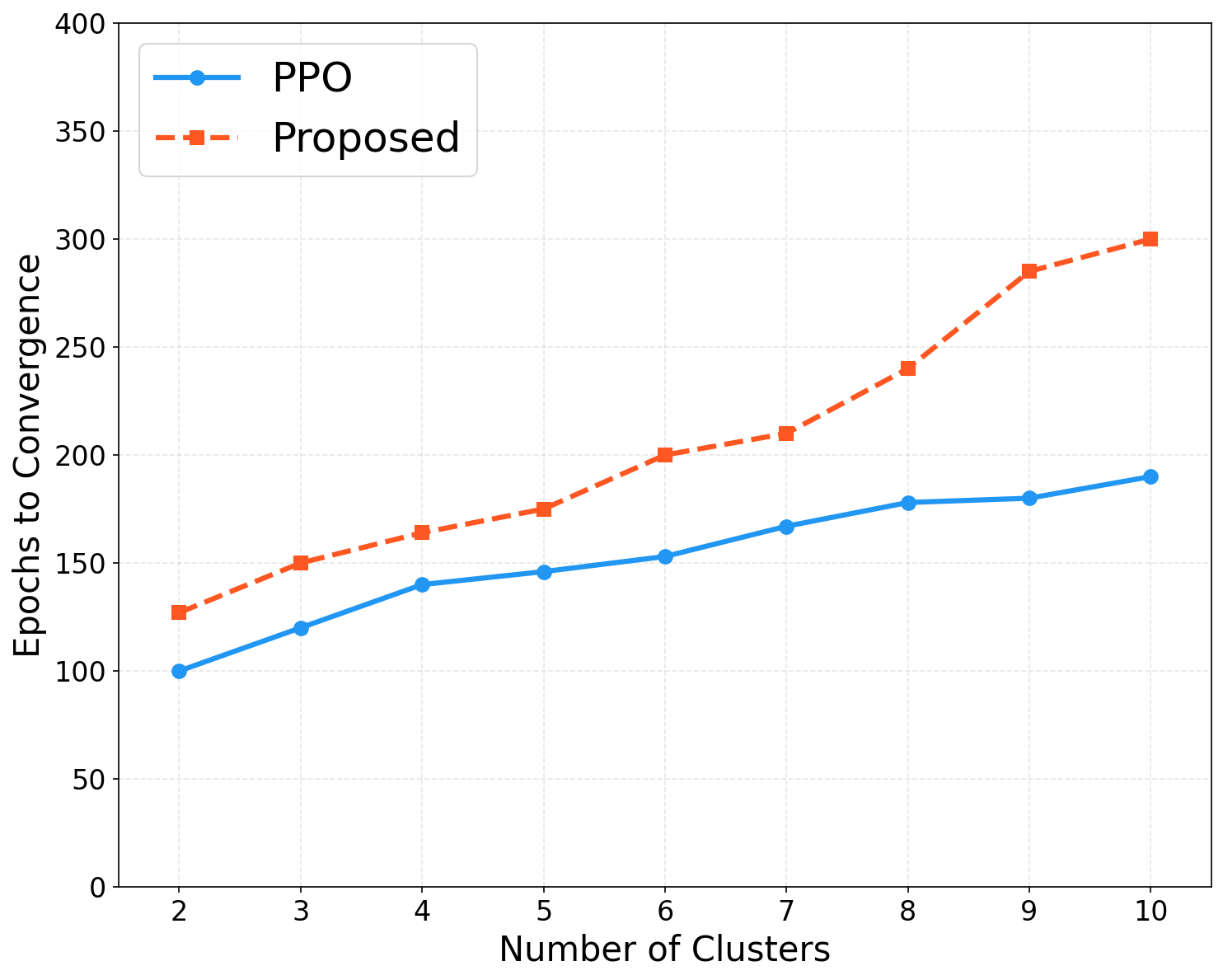}
  \caption{\footnotesize{Scalability of the proposed algorithm and PPO.}}
  \label{fig:convergence_epochs}
  \centering
  \vspace{-0.55cm}
\end{figure}

Fig. \ref{fig:convergence_epochs} visualizes the scalability of the proposed algorithm and the PPO algorithm. From this figure, we can see that both algorithms have slower convergence with a larger number of user clusters. This can be attributed to the fact that the complexity of the DBS trajectory optimization problem increases significantly with the higher number of user clusters. More specifically, we also observe that the proposed algorithm converges 2.36 times more slowly as the number of clusters changes from 2 to 10, which is 46\% higher than the one of the PPO algorithm. This is because the proposed algorithm additionally trains a discriminator and performs adversarial imitation, whereas the PPO algorithm only conducts direct actor-critic policy optimization. As the proposed algorithm offloads the training process to the digital side of the drone network, the complexity is commendable, while the compromised computation cost can contribute to significant improvements in optimality, scalability, and reliability of the drone network, even without direct access to DBS intents. Overall, these make the proposed method an appropriate solution for the intent-based DBS trajectory optimization problem.

% \begin{figure}
% % \setlength{\belowcaptionskip}{-2pt}
% % \setlength{\abovecaptionskip}{-2pt} 
%   \centering
%   \includegraphics[width=0.36\textwidth]{snapshot.eps}
%   %\renewcommand{\captionlabelfont}{\footnotesize}
%   \caption{\footnotesize{Snapshot of trajectories resulting from all considered algorithms.}}
%   \label{Fig. 3}
%   \centering
%   %\vspace{-0.7cm}
% \end{figure}

% \begin{figure}
% % \setlength{\belowcaptionskip}{-6pt}
% % \setlength{\abovecaptionskip}{-2pt} 
%   \centering
%   % \vspace{0.7cm}
%   \includegraphics[width=0.36\textwidth]{converge_largeNN_enlarge.eps}
%   \caption{\footnotesize{Convergence of all considered algorithms.}}
%   \label{Fig. 2}
%   \centering
%   % \vspace{-0.1cm}
% \end{figure}

\section{Conclusion}
%\vspace{-0.1cm}
In this paper, we designed the trajectory for an intent-driven DBS serving distributed wireless ground users under environmental dynamics. We proposed a DT framework that monitors and predicts environmental changes to automate the adjustment of DBS trajectory, despite the absence of an accurate DBS intent modeling. A generative IRL based algorithm is proposed to actuate the automated trajectory adjustment by generating and imitating historical DBS trajectories. Simulation results have shown that the proposed algorithm achieves near-optimal performance with less than 2\% difference with the on-board DBS control with full access to DBS intents, while yielding up to 137\% higher DBS utility than that of on-board DBS control with delayed perception. It is also shown that the proposed algorithm reduces performance loss under environmental changes while maintaining performance comparable to the PPO algorithm. Meanwhile, the training process of the proposed algorithm is implemented over the synthetic DT system, and does not introduce extra hardware and energy costs to the DBS.

\vspace{-0.1cm}

\bibliographystyle{IEEEtran}
\def\baselinestretch{1}
\bibliography{references}

\appendices

\section{The proof of Theorem 1}
As shown in Fig. \ref{Fig. proof1}, the DBS's perception of network states will decide the next network state, which in turn decides the reward. At step $k$ of the DBS trajectory, the true physical network properties are captured in $\boldsymbol{s}_k$, the maneuvering decision of DBS is described as $a_k$, and the perception error on DBS's perception of $\boldsymbol{s}$ is represented $\Delta \boldsymbol{s}_k$. The expected DBS utility under accurate physical network property perception is 
\begin{equation}\label{eq:true state}
J(\boldsymbol{s})=\sum_{\boldsymbol{\xi}\in\mathcal{E}}G(\boldsymbol{\xi})\bigg[{\pi}_d({\boldsymbol{s}}_{k+1}|\boldsymbol{s}_k)\prod_{i \neq k}^K{\pi}_d({\boldsymbol{s}}_{i+1}|\boldsymbol{s}_i)\bigg],
\end{equation}
in which $\pi_d({\boldsymbol{s}}_{k+1}|\boldsymbol{s}_k)=\sum_{a_k}\pi(a_k|\boldsymbol{s}_k)P({\boldsymbol{s}}_{k+1}|\boldsymbol{s}_k,a_k)$, with $\pi(a_k|\boldsymbol{s}_k)$ capturing probability of selecting $a_k$ at $\boldsymbol{s}_k$ and $P(\boldsymbol{s}_{k+1}|\boldsymbol{s}_k,a_k)$ denotes that the DBS moves from $\boldsymbol{s}_k$ to $\boldsymbol{s}_{k+1}$ by taking action $a_k$. Under inaccurate property perception at step $k$, the expected utility becomes
\begin{equation}\label{eq:error state}
    J(\boldsymbol{s}+\boldsymbol{\Delta s}_k)\!=\!\sum_{\boldsymbol{\xi}\in\mathcal{E}}G(\boldsymbol{\xi})\bigg[{\pi}_d({\boldsymbol{s}}_{k+1}|\boldsymbol{s}_k\!+\!\boldsymbol{\Delta s}_k)\!\prod_{i\neq k}^K\!{\pi_d}({\boldsymbol{s}}_{i+1}|\boldsymbol{s}_i)\bigg],
\end{equation}
in which $\pi_d({\boldsymbol{s}}_{k+1}|\boldsymbol{s}_k+\boldsymbol{\Delta s}_k)=\sum_{a_k}\pi(a_k|\boldsymbol{s}_k+\boldsymbol{\Delta s}_k)P({\boldsymbol{s}}_{k+1}|\boldsymbol{s}_k,a_k)$.

Since DBS always chooses the maneuvering decision that maximizes its hit rate at current physical network property, the decisions under perception errors, i.e., $\pi(a_k|\boldsymbol{s}_k+\boldsymbol{\Delta s}_k)$, deviates from the optimal ones in \eqref{eq:true state}, i.e., $\pi(a_k|\boldsymbol{s}_k)$. In other words, when the perception of the physical network property is erroneous, the DBS will choose non-optimal decisions with high probability, while choosing the optimal one with low probability, resulting in a lower probability assigned to high-utility trajectories and thus a reduced expected utility. Therefore,
\begin{equation}
    J(\boldsymbol{s})\ge J(\boldsymbol{s}+\boldsymbol{\Delta s}_k).
\end{equation}
This completes our proof.

 \begin{figure}
  \centering
  %\vspace{-0.1cm}
  \includegraphics[width=0.4\textwidth]{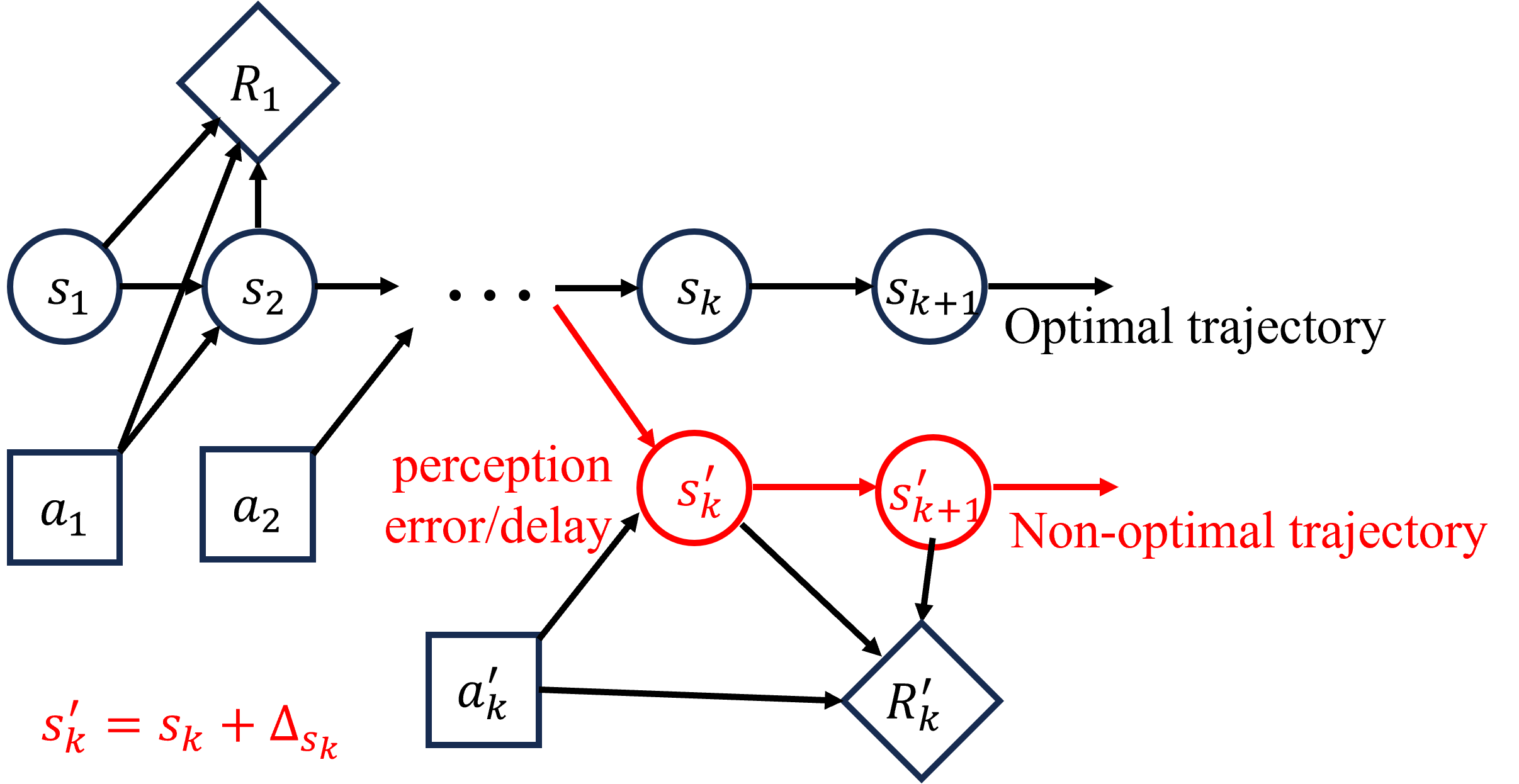}
  \caption{\footnotesize{Illustration of the state misperception.}}
  \label{Fig. proof1}
  \centering
  \vspace{-0.5cm}
\end{figure}

\section{The proof of Theorem 2}
Let $\boldsymbol{\xi}^*$ describe the optimal trajectory that is selected under accurate physical network property perception. From \eqref{eq:true state}, the probability of selecting $\boldsymbol{\xi}^*$ under the accurate physical network property perception is given by
\begin{equation}
    P(\boldsymbol{\xi}^*)\!=\!\prod_{k=1}^K \!\pi(a_k^*|\boldsymbol{s}_k)\! P({\boldsymbol{s}}_{k+1}|\boldsymbol{s}_k, a_k^*).
\end{equation}
Here, $a_k^*$ denotes the optimal action of DBS at step $k$. The probability of selecting $\boldsymbol{\xi}^*$ with a perception error $\boldsymbol{\Delta s}_k$ at step $k$ is given by 
\begin{align}
    &\hat{P}(\boldsymbol{\xi}^*) =\notag \\
    &\pi(a_k^*|\boldsymbol{s}_k\!+\!\boldsymbol{\Delta s}_k)P({\boldsymbol{s}}_{k+1}|\boldsymbol{s}_k,a_k^*)\!\prod_{i \neq k}^K \!\pi(a_i^*|\boldsymbol{s}_i)\! P({\boldsymbol{s}}_{i+1}|\boldsymbol{s}_i, a_i^*).
\end{align}
The logarithmic difference between $P(\boldsymbol{\xi}^*)$ and $\hat P(\boldsymbol{\xi}^*)$ is given by 
\begin{align}
    \log \frac{\hat P(\boldsymbol{\xi}^*)}{P(\boldsymbol{\xi}^*)}\!&=\!\log \frac{\pi(a_k^*\!|\boldsymbol{s}_k\!+\!\boldsymbol{\Delta s}_k)}{\pi(a_k^*\!|\boldsymbol{s}_k)}\!+\!\log \frac{P({\boldsymbol{s}}_{k+1}\!|\boldsymbol{s}_k,a_k^*)}{P({\boldsymbol{s}}_{k+1}\!|\boldsymbol{s}_k,a_k^*)}\notag\\
    &\ge -L_{\pi}\epsilon,
\end{align}
in which $L_{\pi}$ is the local Lipschitz constant for the policy, satisfying $|\log \pi(a_k^*|\boldsymbol{s}_k)-\log \pi(a_k^*|\boldsymbol{s}_k+\boldsymbol{\Delta s}_k)|\le L_{\pi}\epsilon$ and $||\boldsymbol{\Delta s}_k||\le\epsilon$. Therefore, we obtain
\begin{equation}
    \frac{\hat P(\boldsymbol{\xi}^*)}{P(\boldsymbol{\xi}^*)}\ge e^{-L_{\pi}\epsilon}.
\end{equation}
That is, the marginal probability of the optimal trajectory changes on the exponential scale of perception error $\Delta \boldsymbol{s}$, so does the expected utility. This completes our proof.

\end{document}